\documentclass[12pt,letterpaper,twoside]{article}
\usepackage[T1]{fontenc}
\usepackage[english]{babel}
\usepackage[left=1.25in,right=1.25in,top=1.56in,bottom=1.25in,
            headheight=15pt,headsep=10pt]{geometry}
\usepackage{amsmath,amssymb,amsthm,bm}
\usepackage{float,enumitem,graphicx,subcaption,tikz}
\usepackage{microtype,titling,titlesec,fancyhdr}
\usepackage[authoryear,round,comma]{natbib}
\usepackage[colorlinks=true,linkcolor=blue,citecolor=blue,urlcolor=blue]{hyperref}

\renewcommand{\footnoterule}{\kern-3pt\hrule width 4em\kern2.6pt}
\pretitle{\begin{center}\normalfont\normalsize}
\posttitle{\par\end{center}}
\preauthor{\begin{center}\normalsize\begin{tabular}[t]{c}}
\postauthor{\end{tabular}\par\end{center}\vspace{-1.7em}}
\predate{}\postdate{}
\newcommand{\authorname}[1]{\textsc{#1}}
\newcommand{\affiliation}[1]{{\footnotesize #1}}
\newcommand{\keywords}[1]{\par\medskip\noindent\textsc{Keywords:} #1}
\renewenvironment{abstract}
  {\begin{list}{}{\setlength{\leftmargin}{0.14\textwidth}%
                  \setlength{\rightmargin}{\leftmargin}}\item\relax\footnotesize}
  {\end{list}}

\fancypagestyle{plain}{\fancyhf{}}
\newcommand{\sectionlabel}{\thesection.}
\AddToHook{cmd/appendix/after}{\renewcommand{\sectionlabel}{APPENDIX~\thesection:}}
\titleformat{\section}{\centering\normalfont\footnotesize}{\sectionlabel}{0.6em}{\MakeUppercase}
\titleformat{\subsection}{\centering\normalfont\normalsize}{\thesubsection.}{0.5em}{\itshape}
\titleformat{\paragraph}[runin]{\normalfont\itshape}{}{0pt}{}
\titlespacing{\section}{0pt}{1.0\baselineskip}{0.35\baselineskip}
\titlespacing{\subsection}{0pt}{0.8\baselineskip}{0.3\baselineskip}

\newtheoremstyle{result}{\medskipamount}{\medskipamount}{\itshape}
  {\parindent}{\scshape}{:}{0.5em}{\thmname{#1}\thmnumber{ #2}\thmnote{\normalfont---#3}}
\newtheoremstyle{defined}{\medskipamount}{\medskipamount}{\normalfont}
  {\parindent}{\scshape}{:}{0.5em}{\thmname{#1}\thmnumber{ #2}\thmnote{\normalfont---#3}}
\theoremstyle{result}
\newtheorem{thm}{Theorem}
\newtheorem{lem}{Lemma}
\newtheorem{prop}{Proposition}
\newtheorem{ax}{Axiom}
\theoremstyle{defined}
\newtheorem{defn}{Definition}

\usetikzlibrary{calc}
\tikzset{
  trust diagram/.style={
    scale=6,
    x={(0.924cm,-0.514cm)},
    y={(1.32cm,0.128cm)},
    z={(0,0.943cm)},
    every node/.style={font=\small},
    dot/.style={circle,fill,inner sep=1.2pt}
  },
  bias diagram/.style={
    scale=2.8,
    every node/.style={font=\tiny},
    lbl/.style={font=\tiny,fill=white,fill opacity=0.75,
    text opacity=1,inner sep=0.6pt},
    dot/.style={circle,fill,inner sep=0.8pt}
  }
}

\begin{document}

\title{\MakeUppercase{Belief Updating without Complete Trust}\thanks{Noor is at the Department of Economics, Boston University, 270 Bay
State Road, Boston MA 02215. Email: jnoor@bu.edu. Yang is at the Department
of Economics, University of Michigan, 4330 North Quad, 105 S State
Street, Ann Arbor, MI 48109. Email: yuzyang@umich.edu. We are grateful to audiences at Decisions:
Theory, Experiments and Applications Conference 2026, Science of Decision
Making Conference 2026, Boston Theory Conference 2026, Boston University,
University of Michigan, University of Pennsylvania. The usual disclaimer
applies.}}
\author{\authorname{Jawwad Noor}\\
  \affiliation{Department of Economics, Boston University}\\[1.1em]
  \authorname{Yuzhao Yang}\\
  \affiliation{Department of Economics, University of Michigan}}
\date{}
\maketitle

\begin{abstract}
We represent a non-Bayesian agent as one who does not completely trust
the accuracy of the information they receive. A key behavioral expression
of complete trust lies in a homogeneity property of Bayesian updating:
posterior beliefs do not change if a signal is made arbitrarily rare
by scaling down its likelihood vector. We show that simply dropping
this property and retaining all other Bayesian behavioral properties
yields a unique representation where the agent is still Bayesian but
has subjective uncertainty over the information structure generating
the signal. The representation result is proved using the Fundamental
Theorem of Projective Geometry. We analyze how various updating biases
may be rationalized by a lack of trust.

\end{abstract}

\section{Introduction}

The information presented to an agent and the agent's reading of the
information may not always coincide. While there may well be psychological
reasons for this -- perhaps the agent did not put in all the effort
required to comprehend the information, or maybe the framing of the
information impacted their perception of it -- even fully rational
agents will exhibit a gap if they wonder whether the provider of information
possesses hidden motivations and incentives that cause them to lie.
In a world where media slant, misinformation and fake news are often
used to manipulate beliefs in both economic and political environments,\footnote{See \citet{GentzkowShapiro2006} and \citet{AllcottGentzkow2017} for
economic models of media bias and fake news. See \citet{LiuMoss2025}
for the impact of corporate misinformation and fake news in financial
markets.} and where even academic research faces a replicability crisis,\footnote{In psychology, the \citet{OpenScienceCollaboration2015} was unable to
replicate 64\% of the 100 studies from three major psychology journals
they investigated. In economics, \citet{CamererEtAl2016} were unable
to replicate 39\% of the 18 experimental studies published in two
major economics journals.} it may even be irrational to simply take information at face value.

In this paper we explore under what conditions an agent's nonstandard
behavior -- in the sense of non-Bayesian updating with respect to
some information structure $\sigma$ -- can be rationalized by Bayesian
updating with respect to some subjective information structure $\pi$
that misspecifies $\sigma$. A literature (reviewed below) explores
this question by asking whether some rationalizing $\pi$ exists,
a criterion that places few restrictions on $\pi$ itself. Relatedly,
\citet{BohrenHauser2025} observe that ``the line between {[}non-Bayesian
updating{]} and model misspecification is ... fuzzy.'' We deviate
from the existing literature by demanding structure on the agent's
$\pi$: a plausible story for why the agent might subjectively believe
that the information is $\pi$ rather than $\sigma$. Making such
a demand has methodological value: there ``should'' be discipline
on the analyst's ability to explain away observed violations of incumbent
models. Moreover, the hypothesis that the agent is rational despite
their non-Bayesian behavior may have implications testable using
additional data, which should be made explicit.
The story we tell is one where the agent harbors uncertainty about
whether $\sigma$ is indeed the true information structure, possibly
arising from lack of trust, and the paper is devoted to uncovering
the behavioral content of such a story.

In a standard setting of choice under uncertainty with some state
space $\Omega$, our primitive consists of how the agent ranks acts
before and after receiving a signal $s\in S$ from some information
structure (henceforth \emph{experiment}) $\sigma$ that varies in
some primitive set of experiments $\Sigma$. {While it is common
to assume that $\sigma$ is an objective experiment, in practice $\sigma$
is } \textit{communicated} by some information source, and one decides
to what extent to accept what is being communicated. Accordingly,
we interpret $\sigma$ as a ``stated'' experiment. Our first main
result axiomatizes Bayesian updating with complete trust, where each
preference admits a Subjective Expected Utility representation and
where posterior beliefs are Bayesian with respect to $\sigma$: 
\[
p_{\sigma}(\omega|s)=\frac{\sigma(s|\omega)p(\omega)}{\sum_{\omega'\in\Omega}\sigma(s|\omega')p(\omega')},\qquad\omega\in\Omega.
\]
We observe that a key behavioral expression of ``complete trust in
$\sigma$'' lies in a \emph{Homogeneity} condition, which states
that posterior beliefs conditional on a signal $s$ depend only on
relative likelihoods $\text{\ensuremath{\frac{\sigma(s\mid\omega)}{\sigma(s\mid\omega')}}}$.  When the agent is uncertain
about $\sigma$, absolute likelihoods may also matter. If $\sigma$ assigns
little probability to $s$ in every state, it offers a poor explanation for
why $s$ occurred. The agent may then give more weight to an alternative
model under which $s$ is more likely, an intuition familiar from \cite{Ortoleva2012} and \cite{Ba2026}. For instance, if one receives an email from a work
acquaintance asking for money, the sheer improbability of such a request
might trigger the suspicion that the email is a scam. Homogeneity violations have also been documented empirically. Even at a fixed prior and likelihood ratio, making the observed signal proportionally rarer under both states changes average posterior beliefs \citep{saeed2026decomposing}.

We explore how the agent may be represented when Homogeneity is simply
dropped while maintaining all the other axioms. Our second main result
shows that posterior beliefs given $\sigma$ are represented by 
\begin{align}\label{eq00}
p_{\sigma}(\omega|s)=\frac{\left[\tau\sigma(s|\omega)+(1-\tau)\rho(s|\omega)\right]p(\omega)}{\sum_{\omega'\in\Omega}\left[\tau\sigma(s|\omega')+(1-\tau)\rho(s|\omega')\right]p(\omega')},\qquad\omega\in\Omega.
\end{align}
That is, beliefs must be Bayesian with respect to the subjective experiment $\tau\sigma+(1-\tau)\rho$, which deviates from $\sigma$ depending on the degree of trust $\tau$ the agent has in it. 

Under this representation, the observed signal can change both what the agent believes about the state and how much she trusts the stated experiment $\sigma$. The posterior in (\ref{eq00}) is equivalent to a representation (eq.~\eqref{eqmodelavg}, Section~\ref{sec31}) in which the DM updates
$\tau$ by Bayes' rule, shifting weight toward the experiment under which
the observed signal was more likely. She then uses the updated weights
to average the Bayesian posteriors under $\sigma$ and $\rho$.\footnote{In Bayesian statistics, this procedure is known as Bayesian model averaging. For a range of applications, see \citet{hoeting1999}.}  
Thus, an observation can weaken the agent's trust in the stated experiment when the alternative offers a better explanation for it.  In this way, Bayesian updating can give rise to
violations of Homogeneity. A lack of trust in the stated information is therefore behaviorally distinct from existing
non-Bayesian updating rules (reviewed below), which retain Homogeneity and weaken other
normative properties of Bayesian updating.  Our results are proved using the Fundamental
Theorem of Projective Geometry.\footnote{To our knowledge, the Fundamental Theorem has not previously been applied in economics.}

We conclude this Introduction below with related literature. Section
2 provides the foundations of the Bayesian model and pinpoints Homogeneity
as a key behavioral expression of complete trust. Section 3 proves
a representation theorem for the axioms of the Bayesian model excluding
Homogeneity. Section 4 provides the technical insight behind the proof
of the representation theorem. Section 5 studies special cases of
the model, relating both to weakenings of Homogeneity and to experimental
evidence of updating biases. Section 6 shows that the model endogenously
gives rise to a form of confirmation bias, and it provides formal
characterizations while also relating to the evidence. All proofs are relegated to appendices.

\paragraph*{{Related Literature.}}

This paper sits at the intersection of the literatures on belief
updating with model misspecification and on non-Bayesian
updating; our characterization of Homogeneity speaks to the distinction
between the two, as we discuss in Sections 2 and 3. The classic reference is \citet{Ghirardato2002}, who works with dynamic Subjective Expected Utility preferences
over acts and shows that, under regularity conditions and partitional
information, Bayesian updating is equivalent to a dynamic
consistency condition (see Appendix~\ref{Appendix: Classic Dynamic SEU}
for a review). Adapted to our setting his state space is $\Omega\times S$,
so the experiment is subjective and recovered from preference; ours
is objectively given, which allows us to characterize when the Bayesian
agent's subjective model coincides with (or systematically deviates
from) it. See Section 2.2 for a more detailed comparison. 

Taking priors
and posteriors as observable (that is, taking a belief-theoretic rather
than a decision-theoretic approach), \citet{ShmayaYariv2016} show
that a belief sequence is consistent with Bayes' rule whenever the
prior lies in the relative interior of the convex hull of the posteriors.
\citet{BohrenHauser2023} additionally take the agent's forecasts of
posterior distributions as observable, and characterize when beliefs
and forecasts are jointly consistent with Bayesian updating under
a possibly misspecified model. Taking the empirical distribution of
posteriors as observable, \citet{Molavi2025} shows that beliefs are consistent
with Bayesian updating if and only if the mean posterior is absolutely
continuous with respect to the prior. \citet{FudenbergLanzani2026}
characterize when one-step-ahead forecasts of signal realizations
are consistent with Bayesian learning about a stable but unknown i.i.d.\ signal-generating
process. This literature asks whether some rationalizing $\pi$ exists.
We ask instead what behavior looks like when $\pi$ reflects the agent's
uncertainty about whether the presented experiment is the one generating
her signals.

Next, we relate our work to the non-Bayesian updating literature.
\citet*{EpsteinNoorSandroni2008} provide a general framework by
relating non-Bayesian updating to a preference for commitment. \citet{Ortoleva2012} proposes a hypothesis-testing model in which the agent applies
Bayes' rule after anticipated events but reconsiders her prior upon
sufficiently unlikely observations. The intuition that rare signals
prompt reconsideration is shared with our model, though the response
differs: his agent revises her prior, whereas ours remains Bayesian
with respect to an enlarged model space. \citet{Cripps2018}, \citet*{ChambersMasatliogluRaymond2024}, \citet{Jakobsen2026}, and \citet{Whitmeyer2026}
characterize generalizations of Bayes' rule satisfying distinct normative
criteria. Each retains the Homogeneity axiom that we drop, so the
non-Bayesian agent they describe is behaviorally distinct from ours
-- see Section 2.2 for a more detailed discussion. \citet{Kovach2021},
\citet{Yang2025}, and \citet{YangZhao2026} study various systematic deviations
from Bayesian updating, including conservatism, overprecision, and
representativeness heuristics. \citet{NoorPayro2026} formalize the
law of small numbers as a belief in mean reversion. These papers take
the departure itself as the object of axiomatization, whereas ours
derives it from uncertainty about the source of the signal. \citet{Zhao2022},
\citet*{KeWuZhao2024}, \citet{MaZhao2024}, and \citet*{DominiakKovachTserenjigmid2025} axiomatize updating rules for qualitative
statements of relative likelihood and for sets of probability distributions.
These take a different form of information as the primitive, rather
than a different rule for updating on an experiment.

Finally, we relate our work to the small literature that studies uncertainty
about information. \citet{GentzkowShapiro2006} study media bias in
a model where agents are uncertain about the quality of a news source,
and they use Bayesian updating exactly as in our paper. Without presenting
formal results, they note that some form of confirmation bias arises
endogenously in such a model -- we flesh out their intuition formally
in Section \ref{sec:Prior-Dependence}. \citet{ChengHsiaw2022} similarly
adopt Bayesian updating to evaluate the uncertainty (which they call
``distrust'') about the quality of experts, although their agents
subsequently update beliefs about the state in a non-Bayesian fashion.
While our paper works with posteriors after a single signal, \citet{Hu2026}
works with posteriors after sequences of signals of different lengths
and axiomatically characterizes Bayesian updating with complete certainty
(which he calls ``confidence'') in the information structure. In
an experiment, he shows that 95\% of subjects behave as though they
are not completely confident.

\section{Bayesian Updating and Trust}

We characterize Bayesian updating and identify the behavioral
content of complete trust.

\subsection{Primitives}

Consider a finite state space $\Omega$ with $|\Omega|\geq4$ and
a finite signal space $S$. An experiment is a mapping $\sigma:\Omega\rightarrow\Delta(S)$.
For each $s\in S$, let $\sigma_{s}:=\bigl(\sigma(s\mid\omega)\bigr)_{\omega\in\Omega}\in[0,1]^{\Omega}$
denote the likelihood vector of signal $s$ under experiment $\sigma$.
We write an experiment as a $|\Omega|\times|S|$ matrix with states
as rows and signals as columns. For experiments
$\sigma,\sigma'$ and $\alpha\in[0,1]$, the mixture $\alpha\sigma+(1-\alpha)\sigma'$
is the experiment with likelihood vector $\alpha\sigma_{s}+(1-\alpha)\sigma'_{s}$
at each $s$.

Consider a nonempty, convex set of experiments $\Sigma$. We assume
that, for each $s\in S$, the induced set of likelihood vectors 
$
\Sigma_{s}:=\{\sigma_{s}\mid\sigma\in\Sigma\}\subset\mathbb{R}^{\Omega}
$
is open in $\mathbb{R}^{\Omega}$ and contains $0$ in its closure.
Appendix~\ref{sec:relaxing-boundary} shows that the boundary condition
$0\in\overline{\Sigma_{s}}$ can be replaced by an additional axiom,
yielding the same representation below.
Thus the essential assumptions are convexity and openness.
In particular, each $\Sigma_{s}$ may well be an arbitrarily small
open convex neighborhood of a given (interior) likelihood vector.

We interpret $\Sigma$ as the set of experiments that the agent may
be presented with by a given sender. For instance, if the signal will
be a ball drawn from an urn containing balls with colors in a set
$S$, and if any non-degenerate composition of the urn could have
been adopted by the sender, then $\Sigma$ is the set of all interior
experiments. The set $\Sigma$ may be more restricted, for example,
when language or convention or even a known bias possessed by the sender
qualitatively restricts how signals correspond to states. For instance,
by convention, an enthusiastic letter of recommendation endorses
the candidate's quality, although the strength of the endorsement
depends on what the writer may or may not say between the lines, leading
to some restricted set $\Sigma$ as described in \eqref{eqdefSigma}
below. While $\Sigma$ is tied to the identity of a sender, in our
model the sender also ``presents'' the agent with some $\sigma\in\Sigma$.
The sender's selection process is unmodelled. 

Let $X$ be a convex subset of a normed vector space, and denote the space of acts by $\mathcal{F}=\{f:\Omega\rightarrow X\}$.
Our primitive is the family of binary relations over $\mathcal{F}$
given by $\{\succsim_{0},\succsim_{\sigma,s}\}_{(\sigma,s)\in\Sigma\times S}$,
capturing how the agent ranks acts before receiving
any information, and after being told that signal $s$ was generated
by experiment $\sigma$.\footnote{Our axioms do not use the vector space structure of $X$ directly. It would be used explicitly if Axiom \ref{axiom:seu}
is replaced with the Anscombe-Aumann axioms. Our proofs rely on $u(X)$ being a nontrivial interval, which is guaranteed
by the vector space structure and the assumption in Axiom \ref{axiom:seu}
that $u$ is non-constant and affine.} We maintain throughout that

\setcounter{ax}{-1} \begin{ax}[SEU]\label{axiom:seu} For
some non-constant affine $u:X\rightarrow\mathbb{R}$ and full-support
distribution $p\in\Delta(\Omega)$, $\succsim_{0}$ admits a representation
\[
U(f)=\int_{\Omega}u\circ f\,dp,\qquad f\in\mathcal{F},
\]
and for each $(\sigma,s)\in\Sigma\times S$, $\succsim_{\sigma,s}$
admits a representation 
\[
U_{\sigma,s}(f)=\int_{\Omega}u\circ f\,dp_{\sigma}(\cdot\mid s),\qquad f\in\mathcal{F},
\]
for some distribution $p_{\sigma}(\cdot\mid s)\in\Delta\Omega$. \end{ax}

Therefore, there is a prior $p(\cdot)$ underlying $\succsim_{0}$
and a posterior $p_{\sigma}(\cdot\mid s)$ conditional on $(s,\sigma)$
underlying $\succsim_{\sigma,s}$. Our results
derive a representation for $p_{\sigma}$ via axioms
on $\{\succsim_{0},\succsim_{\sigma,s}\}_{(\sigma,s)\in\Sigma\times S}$.
In this section we seek axioms such that for any experiment
$\sigma:\Omega\rightarrow\Delta(S)$ and signal $s\in S$, the posterior
$p_{\sigma}(\cdot\mid s)$ is the \emph{Bayesian posterior}, which
we denote as 
\[
p^{*}_{\sigma}(\omega\mid s)=\frac{\sigma(s\mid\omega)p(\omega)}{\sum_{\omega'\in\Omega}\sigma(s\mid\omega')p(\omega')},\qquad\omega\in\Omega.
\]

\subsection{Foundations}

\label{sec22}

\noindent In the following axioms, we omit the quantifiers ``for
all $\sigma,\sigma',\sigma''\in\Sigma$,'' ``for all $s\in S$,''
``for all $f,g,h\in\mathcal{F}$,'' and ``for all $x,y\in X$.''
Our main axiom is

\begin{ax}[Informational STP]\label{a1} If $\sigma''=\alpha\sigma+(1-\alpha)\sigma'$
for $\alpha\in(0,1)$, then 
\[
f\succsim_{\sigma,s}g\text{ and }f\succsim_{\sigma',s}g\implies f\succsim_{\sigma'',s}g.
\]
\end{ax}

{Imagine that the agent receives signal $s$ but does not know yet
which experiment they are facing: they know only that with probability
$\alpha$ the signal $s$ is generated by $\sigma$, and with the
remaining probability by $\sigma'$. Maintaining
the ``background assumption'' that the agent reduces lotteries,
this lottery corresponds to facing the experiment $\alpha\sigma+(1-\alpha)\sigma'$.
The axiom states that if the agent would prefer $f$ over $g$ regardless
of whether the experiment is known to be $\sigma$ or $\sigma'$,
then they must prefer $f$ over $g$ when facing $\alpha\sigma+(1-\alpha)\sigma'$,
that is, before finding out which experiment it is. The axiom
is thus closely related to the sure-thing principle of \citet[p.~21]{Savage1954}: ``If the person would not prefer $f$ to $g$, either knowing
that the event $B$ obtained, or knowing that the event ${\sim}B$
obtained, then he does not prefer $f$ to $g$.'' This normatively
compelling axiom is the notion of dynamic consistency used
in a risk setting \citep[see][]{Machina1989}, applied to uncertainty over experiments: preferences before and after
resolution of the lottery $\alpha\sigma+(1-\alpha)\sigma'$ are consistent
with each other.}\footnote{{It does not correspond to Dynamic Consistency (Definition~\ref{DC},
Appendix~\ref{Appendix: Classic Dynamic SEU}) in a subjective risk
setting, since the statement is not of the form that ``if the agent
would ex ante prefer to choose $f$ over $g$ conditional on signal
$s$, then they respect this ex post after observing $s$''.}}

{ The normative and economic relevance of the axiom lies in the fact
that its violation would open the agent to exploitation: Suppose that
the agent is endowed with act $f$ and satisfies $f\succsim_{\sigma,s}g$,
$f\succsim_{\sigma',s}g$ and $g\succ_{\alpha\sigma+(1-\alpha)\sigma',s}f$.
A bookmaker can present the agent with information $\alpha\sigma+(1-\alpha)\sigma'$
and, if signal $s$ is generated, can offer the agent the opportunity
to exchange $f$ for $g$, which the agent will accept for some small
enough price $\epsilon>0$. But the bookmaker can then resolve the
lottery $\alpha\sigma+(1-\alpha)\sigma'$ and regardless of the outcome,
the agent can be induced to return to the status quo by exchanging
$g$ for $f$, while the bookmaker earns $\epsilon$ from the
agent without any risk of paying a cost.}\footnote{{This corresponds to what \citet{Machina1989} refers to as a ``classical
making-book'' argument. }}

The remaining axioms are more obvious properties of Bayesian updating.
As usual, for any acts $f,h$ and event $E\subset\Omega$, the act
$fEh$ pays $f(\omega)$ for all $\omega\in E$
and $h(\omega)$ for all $\omega\notin E$.

\begin{ax}[Independence of Irrelevant Details]\label{a2} If
$\sigma(s\mid\omega)=\sigma'(s\mid\omega)$ for all $\omega\in E\subset\Omega$,
then 
\[
fEh\succsim_{\sigma,s}gEh\iff fEh\succsim_{\sigma',s}gEh.
\]
\end{ax}

Axiom \ref{a2} embodies two forms of consequentialism. First, when
$E=\Omega$, it requires $\succsim_{\sigma,s}$ to depend only on
the likelihood vector $\sigma_{s}$: signals other than $s$ have not occurred, so how likely they would have been
is irrelevant to the agent's choice. Second, when $E\subsetneq\Omega$,
the evaluation of $fEh$ versus $gEh$ depends only on the likelihoods
assigned to states in $E$: states outside $E$ yield the same consequence
under the two acts, so their likelihoods are irrelevant to the comparison.
Equivalently, the posterior odds between states in $E$ do not depend
on the likelihood of $s$ in states outside $E$. In this
sense, the axiom is also a counterpart of Luce's independence of irrelevant
alternatives, with states outside $E$ as the irrelevant alternatives. Conditions of this form are known to underlie
the characterization of geometric opinion pooling \citep*{GenestWeerahandiZidek1984} and of the power-weighted distortions of beliefs
and experiments in \citet{Grether1980} \citep*{ChambersMasatliogluRaymond2024,Chan2026}.

\begin{ax}[Homogeneity]\label{a3} If $\sigma_{s}=\alpha\sigma'_{s}$
for $\alpha>0$, then $\succsim_{\sigma,s}=\succsim_{\sigma',s}$.
\end{ax}

Homogeneity is the behavioral counterpart of the familiar Bayesian
property that scaling a likelihood vector by a common factor leaves
the posterior unchanged.\footnote{One way to motivate it normatively is to suppose that $s$ denotes
the joint occurrence of an observation $s_{0}$ and Heads from a coin
toss independent of the state and the original signal, with $\Pr(H)=\alpha\in(0,1)$.
The likelihood of $s$ in each state $\omega$ is then $\alpha\Pr(s_{0}\mid\omega)$.
Varying $\alpha$ changes how often $s$ occurs but leaves its likelihood
ratios unchanged. Homogeneity requires these changes to have no effect
on the agent's conditional preferences.}

\begin{ax}[Consistency with Prior]\label{a4} For $\alpha>0$,
if $\sigma_{s}=\sigma'_{s}+\alpha1_{\Omega}$, then 
\[
f\succ_{0}g\text{ and }f\succsim_{\sigma',s}g\implies f\succ_{\sigma,s}g.
\]
\end{ax}

Recall that $\sigma$ is uninformative at signal $s$ if $\sigma_{s}$
is a constant vector: the ratios of entries satisfy $\frac{\sigma(s\mid\omega)}{\sigma(s\mid\omega')}=1$
for all $\omega,\omega'$. Note also that, for an arbitrary $\sigma$,
adding a constant to its likelihood vector $\sigma_{s}$ makes the
signal $s$ less informative, since it brings all the ratios closer
to $1$, that is, $\left|\frac{\sigma(s\mid\omega)}{\sigma(s\mid\omega')}-1\right|\geq\left|\frac{\sigma(s\mid\omega)+c}{\sigma(s\mid\omega')+c}-1\right|$
for any $c>0$. Therefore, the axiom begins with $\sigma_{s},\sigma'_{s}$
where $\sigma'_{s}$ is more informative than $\sigma{}_{s}$. The
axiom states that if prior to any signal the agent preferred $f\succ_{0}g$
and if they still exhibited $f\succsim_{\sigma',s}g$ with an informative
likelihood vector $\sigma'_{s}$, then they would exhibit $f\succ_{\sigma,s}g$
with a less informative likelihood vector $\sigma_{s}$.  Intuitively,
less informative signals push posteriors towards the prior. This expresses
that ex post beliefs maintain a connection with prior beliefs. This
is unlike, for instance, \citet*{EpsteinNoorSandroni2008} where the agent retroactively
changes their prior.\footnote{While appearing related, the axiom is logically independent of Informational
STP (Axiom~\ref{a1}). Axiom~\ref{a1}, and indeed Axioms~\ref{a1}--\ref{a3}
together, impose no restriction relating $\succsim_{\sigma,s}$ to
the ex ante ranking $\succsim_{0}$, and so cannot deliver Axiom~\ref{a4}.
Conversely, Axiom~\ref{a4} concerns a single pair of likelihood
vectors differing by a constant, and says nothing about mixtures,
so it cannot deliver Axiom~\ref{a1}.}

For any act $f$ and state $\omega$, consider the constant act that
yields consequence $f(\omega)$ at every state, and denote it by $f_{\omega}$.
Our final axiom is a richness condition:

\begin{ax}[Richness] \label{a5} For all $s\in S$ and $f\in\mathcal{F}$,
if $f_{\omega}\not\sim_{0}f_{\psi}$ for some $\omega,\psi\in\Omega$,
then there exist $\sigma,\sigma'\in\Sigma$ and $x\in X$ such that
\(
f\succ_{\sigma,s}x\text{ and }f\prec_{\sigma',s}x.
\)
\end{ax}

Richness states that, fixing any signal $s$ and any $f$ with distinct
best and worst outcomes, there always exists a consequence $x$ (which,
given Axiom 0, must be ranked strictly between the best and worst
outcomes of $f$) and a pair of experiments $\sigma,\sigma'$ such
that $f$ is better than $x$ at one experiment and worse at another.
This rules out, for instance, an agent who never updates beliefs,
regardless of the experiment.

These axioms are the foundations for Bayesian updating in our domain:

\begin{thm}\label{thmmain1} $\{\succsim_{0},\succsim_{\sigma,s}\}_{(\sigma,s)\in(\Sigma\times S)}$
satisfies Axioms 1 - 5 if and only if $p_{\sigma}(\cdot\mid s)={p^{*}_{\sigma}(\cdot\mid s)}$ for all $(\sigma,s)\in\Sigma\times S$.
\end{thm}

See Section \ref{subsec:Proof outline} for an outline of the proof.
Our primitive consists of preferences over acts on $\Omega$, rather
than over the richer domain $\Omega\times S$ as in the approach of \citet{Ghirardato2002} adapted to our setting (see Appendix~\ref{Appendix: Classic Dynamic SEU}).
This has two advantages. First, it demands less of the analyst, who need not observe choices over acts contingent on both states and signals. Second, the two approaches characterize different things. When acts
are defined on $\Omega\times S$, the experiment is subjective and
recovered from preference, so what is characterized is Bayesian updating
with respect to a purely subjective joint distribution over $\Omega\times S$.
With acts defined on $\Omega$ alone and an experiment $\sigma$ objectively
given, Theorem~\ref{thmmain1} tells us that the agent's posterior
agrees with the Bayesian posterior computed from that experiment.
We thereby characterize not just Bayesian updating, but Bayesian updating
correctly specified relative to the stated experiment.

\noindent \textbf{Remark.} Theorem~\ref{thmmain1} organizes existing models
of non-Bayesian updating by the normative property each gives
up. Table~\ref{tab:axioms} records which of Axioms~\ref{a1}--\ref{a4}
hold under several prominent rules.

\begin{table}[H]
\centering \caption{Axioms satisfied by alternative updating rules}
\label{tab:axioms} %
\begin{tabular}{lcccc}
\hline 
 & STP  & IID  & Homogeneity  & CwP \tabularnewline
 & (\ref{a1})  & (\ref{a2})  & (\ref{a3})  & (\ref{a4}) \tabularnewline
\hline 
Bayesian updating  & \checkmark  & \checkmark  & \checkmark  & \checkmark \tabularnewline
\citet{Grether1980}  & $\times$  & \checkmark  & \checkmark  & $\times$ \tabularnewline
\citet*{EpsteinNoorSandroni2008}  & \checkmark  & $\times$  & \checkmark  & \checkmark \tabularnewline
Affine distortion \citep{Whitmeyer2026}  & \checkmark  & $\times$  & \checkmark  & $\times$ \tabularnewline
Other rules$^{a}$  & $\times$  & $\times$  & \checkmark  & $\times$ \tabularnewline
\hline 
\end{tabular}

\vspace{0.2cm}

\begin{minipage}[c]{0.9\textwidth}%
{\footnotesize\emph{Note:}}{\footnotesize{} STP, IID, and CwP stand
for Informational STP (Axiom~\ref{a1}), Independence of Irrelevant
Details (Axiom~\ref{a2}), and Consistency with Prior (Axiom~\ref{a4}),
respectively. $\times$ indicates that the axiom may fail.}{\footnotesize\par}

{\footnotesize$^{a}$ \citet{Cripps2018}; \citet{Jakobsen2026}. }{\footnotesize\par}%
\end{minipage}
\end{table}

Notably, the one axiom that none of these rules relaxes is Homogeneity. Departures from this axiom have been observed experimentally.\footnote{\cite{saeed2026decomposing} finds that, holding the prior and likelihood ratio fixed, subjects update less on average when the observed signal is rarer under both states.} We next offer a conceptual motivation for relaxing Homogeneity.

\subsection{Homogeneity as Complete Trust}

\label{subsec:Homogeneity violation}

An intuition present in \citet{Ortoleva2012} and \citet{Ba2026} is that observations
that do not fit the agent's current model well may lead her to reconsider
it in favor of an alternative. In particular, if an observation is
very unlikely under that model but less surprising under the alternative, they may lose confidence in it. In our setting, scaling down the likelihood
vector of the observed signal makes the signal less likely ex ante. Thus, if an agent harbors any doubts
about the experiment, those doubts should get stronger on observing
the signal, and this should be reflected in the posterior.\footnote{In the limiting case, the experiment assigns zero probability
to the observed signal in every state. The observation is then
logically inconsistent with the experiment's predictions, so the
agent must turn to an alternative data-generating process to
account for it.} This would
in turn be reflected in the agent's preferences as a violation of
Homogeneity. The Bayesian updater in contrast satisfies Homogeneity
because they do not lose confidence, on account of their complete
trust.

For a concrete example, let $\Omega=\{\text{Need},\text{No Need}\}$
describe whether a given sibling needs financial help. Let $S=\{\text{ask},\text{dont}\}$,
where ``$\text{ask}$'' denotes receiving a text message from the
sibling's number, purportedly from the sibling and asking for money,
and ``$\text{dont}$'' denotes not receiving such a message. The
agent chooses between lending money to the sibling, an act $f$ with
utility payoffs $+1$ in ``$\text{Need}$'' and $-1$ in ``$\text{No Need}$,''
and not lending, an act $g$ with payoff $0$ in both states. Consider
two situations: in the first, the sibling regularly communicates by
text; in the second, the sibling rarely texts, usually preferring
phone calls or FaceTime. The respective experiments are 
\[
\sigma'=\begin{array}{c|cc}
 & \text{ask} & \text{dont}\\
\hline \text{Need} & 2/3 & 1/3\\
\text{No Need} & 1/3 & 2/3
\end{array}\qquad\text{and}\qquad\sigma=\begin{array}{c|cc}
 & \text{ask} & \text{dont}\\
\hline \text{Need} & 0.0002 & 0.9998\\
\text{No Need} & 0.0001 & 0.9999
\end{array}
\]
Thus, $\sigma_{\text{ask}}=0.0003\,\sigma'_{\text{ask}}$. Suppose the agent also entertains the subjective possibility that \textit{the message was sent by a scammer}. This possibility corresponds to an alternative
data-generating process $\rho$ that generates ``$\text{ask}$''
with probability one in every state.

Suppose that, in the first situation, the agent regards the message
as likely to be a genuine request and chooses to lend, exhibiting
$f\succ_{\sigma',\text{ask}}g$. In the second, receiving the same
message from someone who rarely texts may prompt her to consider the
alternative explanation that it is a scam generated by $\rho$. If
this possibility makes her sufficiently doubtful that her sibling
is in need, she may decline to lend, exhibiting $g\succ_{\sigma,\text{ask}}f$,
which violates Homogeneity.\footnote{The response can also go in the opposite direction: an unusual text
from a sibling who normally calls may suggest an emergency, strengthening
the belief that the sibling is in need. Such a response can likewise
violate Homogeneity. Our point is that limited confidence in the stated
experiment can generate a violation, not that it determines the direction
of the violation.}

The example illustrates how latent uncertainty about the signal-generating
process can affect observable choice even when beliefs about that
process are not directly observed. Scaling down the likelihood vector
may reduce the agent's ex post weight on the presented experiment,
but this weight cannot be backed out directly from choices among acts
contingent on $\Omega$. If acts could instead depend jointly on the
state in $\Omega$ and on which of the signal-generating processes
entertained by the agent actually operated, choices among such bets
could reveal her posterior beliefs about those processes. Such a richer
domain, however, may be difficult to observe: the alternative processes
may be subjective explanations unknown to the modeler, and the process
that actually operated may not be verifiable ex post to determine payoffs. With only $\Omega$-contingent acts observable,
violations of Homogeneity offer a plausible behavioral manifestation
of limited confidence in the presented model, even if the agent uses
Bayesian updating over an enlarged state space that includes uncertainty
about the signal-generating process.

\section{Bayesian Updating with Limited Trust}

Motivated by the interpretation of Homogeneity as an expression of
complete trust in the communicated experiment, we now drop this axiom
and ask what the remaining axioms imply for belief updating.\footnote{We continue to assume that the agent fully trusts that her chosen
act will deliver its specified state-contingent consequences. Limited
trust concerns the information, not the implementation of the act.} Lack of complete trust could in principle extend further: an agent
may react to each communicated experiment $\sigma$ in an arbitrary
manner, potentially violating any of Axioms 1, 2 and 4. However we
maintain these axioms on normative grounds: \emph{an agent who harbors
doubts about $\sigma$ may still wish to respect the consequentialism
embodied in Independence of Irrelevant Details, the intuitive connection
to prior beliefs required by Consistency with Prior, and the dynamic
consistency expressed by Informational STP}.

\subsection{Foundations}

\label{sec31}

Begin with some definitions. Say that $\{\succsim_{0},\succsim_{\sigma,s}\}_{(\sigma,s)\in(\Sigma\times S)}$
satisfies \textit{idempotence} if, for some $\sigma^{*}\in\Sigma$,
$
\succsim_{\sigma^{*},s}=\succsim_{0}\text{ for all }s\in S.$\label{eq:Idempotence}
That is, the agent treats some experiment $\sigma^{*}$ in $\Sigma$ as uninformative, with all posteriors equal to the prior. In the Bayesian model
any uninformative experiment would serve as $\sigma^{*}$, but Idempotence
is silent on what $\sigma^{*}$ is here.

Define a \emph{quasi-experiment} as a non-zero mapping $\rho:\Omega\rightarrow\mathbb{R}^{S}_{+}$,
with $\rho(s\mid\omega)\geq0$ for all $s\in S$, and $\sum_{s\in S}\rho(s\mid\omega)\leq1$
for all $\omega\in\Omega$. This weakens the notion of an experiment
by dropping the requirement that $\sum_{s\in S}\rho(s\mid\omega)=1$.
A quasi-experiment $\rho$ can be interpreted as an experiment with
an additional signal $s^{*}$ outside $S$, occurring with probability
$\rho(s^{*}\mid\omega)=1-\sum_{s\in S}\rho(s\mid\omega)$ at each
$\omega\in\Omega$. Under this reading, the agent's perceived signal
space is larger than the analyst's, and the presented experiment $\sigma$
is understood to produce the additional signal $s^{*}$ with probability
zero.

We find that dropping Homogeneity from the Bayesian model yields:

\begin{thm}\label{thmmain2} $\{\succsim_{0},\succsim_{\sigma,s}\}_{(\sigma,s)\in(\Sigma\times S)}$
satisfies Axioms 1, 2, 4, and 5 if and only if there exists a quasi-experiment
$\rho$ and $\tau\in(0,1]$ such that for all $(\sigma,s)\in\Sigma\times S$,
\begin{align}
p_{\sigma}(\cdot\mid s)=p^{*}_{\tau\sigma+(1-\tau)\rho}(\cdot\mid s). \label{eqmain3}
\end{align}
Moreover, there exists a unique maximal trust $\tau^{*}\in(0,1]$
and a quasi-experiment $\rho^{*}$ that satisfies (\ref{eqmain3}).
If $\tau^{*}<1$, then $\rho^{*}$ is also unique. \end{thm}

\begin{thm}\label{thmmain3} Suppose $\{\succsim_{0},\succsim_{\sigma,s}\}_{(\sigma,s)\in(\Sigma\times S)}$
satisfies idempotence. Then Axioms 1, 2, 4, and 5 are satisfied if
and only if there exists an experiment $\rho$ and $\tau\in(0,1]$
such that (\ref{eqmain3}) holds for all $(\sigma,s)\in\Sigma\times S$. Moreover, $\tau$ is unique, and $\rho$ is unique if $\tau<1$. \end{thm}

Thus an agent who harbors uncertainty about the experiment presented
to her---and accordingly violates Homogeneity --- but adheres to
the remaining normative axioms must update by Bayes' rule with respect
to an enlarged model space, entertaining a subjective $\rho$ in addition
to the presented $\sigma$. Uncertainty about the information structure
calls for a revision of the agent's model space, not of the Bayesian
paradigm. Specifically, she should apply Bayes' rule to the mixture $\tau\sigma+(1-\tau)\rho$ rather than the presented
$\sigma$. The mixture is the reduction of a compound experiment
in which each signal $s\in S$ is generated by $\sigma$
with probability $\tau$ and by $\rho$ otherwise, where
$\rho$ is in general a quasi-experiment but under Idempotence
can be taken to be an experiment, and is unique when $\tau<1$.
Accordingly, we read $\tau$ as the agent's \emph{trust} in $\sigma$,
and $\rho$ as her alternative theory of signal generation.

While the agent has prior trust $\tau$ over $\sigma$ versus $\rho$,
it should be appreciated that the agent \emph{updates} this prior
trust after observing a signal. For any experiment $e=\sigma,\rho$,
denote the \textit{ex ante} probability of signal $s$ by $e_{p}(s):=\sum_{\omega'\in\Omega}e(s\mid\omega')p(\omega')$,
and suppose $e_{p}(s)>0$ for $e=\sigma,\rho$.
Grouping the $\sigma$- and $\rho$-terms in the posterior and normalizing
each by its own ex ante signal probability yields the following formula, familiar from the statistical literature on Bayesian model averaging (\citealt{hoeting1999}):

\begin{align}\label{eqmodelavg}
p_{\sigma}(\omega\mid s)=\underbrace{\frac{\tau\sigma_{p}(s)}{\tau\sigma_{p}(s)+(1-\tau)\rho_{p}(s)}}_{\tau(\sigma\mid s)}\,p^{*}_{\sigma}(\omega\mid s)+\underbrace{\frac{(1-\tau)\rho_{p}(s)}{\tau\sigma_{p}(s)+(1-\tau)\rho_{p}(s)}}_{1-\tau(\sigma\mid s)}\,p^{*}_{\rho}(\omega\mid s),
\end{align}
where the belief $p^{*}_{\sigma}(\cdot\mid s)$ (\textit{resp.} $p^{*}_{\rho}(\cdot\mid s)$)
is the Bayesian posterior with respect to $\sigma$ (\textit{resp.}
$\rho$). In this expression, the agent's posterior is a weighted
average of the two Bayesian posteriors, with the weight on experiment
$\sigma$ denoted as $\tau(\sigma\mid s)$. We interpret this term
as the agent's \emph{updated trust} in $\sigma$ after observing $s$:
it is the prior trust $\tau$ revised in light of the signal. This
updated trust depends on the ex ante probability of $s$: if $\tau\in(0,1)$,
then for every $s\in S$, 
\[
\tau(\sigma\mid s)\ge\tau\iff\sigma_{p}(s)\ge\rho_{p}(s),
\]
that is, trust in $\sigma$ rises at signal $s$ if and only if $s$
is more likely under $\sigma$ than under $\rho$ according to the
agent's prior. This is the source of the Homogeneity violation in
the model, and it formalizes the intuition of Section~\ref{subsec:Homogeneity violation}.

\subsection{Interpretation of \texorpdfstring{$\rho$}{rho}}

\label{sec32}

The representation~\eqref{eqmain3} can be viewed as the reduced
form of a richer environment perceived by the agent. We present three
examples.

\paragraph{Example 1: Aggregation of Alternative Theories}

In the representation~\eqref{eqmain3}, $\rho$ is
a single theory of signal generation entertained as an alternative
to $\sigma$. The agent may in fact entertain several such theories.
While the analyst presents her with $\sigma$, she may have been presented
with other experiments $\{\rho_{i}\}_{i\in I}$ by sources the analyst
does not observe, or formed them herself. For instance,
a reader may be told that a news outlet reports according to $\sigma$,
while other commentators claim that it reports according to
$\{\rho_{i}\}_{i\in I}$, reflecting different assessments of the
outlet's editorial slant. Another example is a letter of recommendation.
A letter could be taken at face value, implying a particular mapping
$\sigma$ between the qualities of a candidate and the words used
to describe them. But one could also read between the lines in multiple
ways, corresponding to different such mappings $\{\rho_{i}\}_{i\in I}$.

In such situations, uncertain which theory is correct, the agent may
assign trust $\tau$ to $\sigma$ and trust $\tau_{i}$ to each $\rho_{i}$,
with $\tau+\sum_{i}\tau_{i}=1$, and update her beliefs using the
reduced experiment $\tau\sigma+(1-\tau)\rho$, where $\rho:=\sum_{i\in I}\frac{\tau_{i}}{1-\tau}\,\rho_{i}.$
Her behavior is thus observationally equivalent to the single-theory
case, with $\rho$ being the trust-weighted average of the alternatives.
This resembles opinion pooling \citep{DeGroot1974}, with two differences:
the agent pools experiments rather than priors, and the components
$\rho_{i}$ are not observed by the analyst.

\paragraph{Example 2: Distortions of $\sigma$}

The agent's alternative theory need not be a fixed experiment unrelated
to $\sigma$. She may instead suspect that signals are generated by
a distorted version $\rho(\sigma)$ of the stated experiment, and form subjective experiment $\pi\sigma+(1-\pi)\rho(\sigma)$.
Whenever the distortion takes the following form  for some experiment $\rho$:
\begin{align}
\rho(\sigma)=\lambda\sigma+(1-\lambda)\rho,\qquad\lambda\in(0,1),\label{eqfixed}
\end{align}
 her subjective experiment reduces to~\eqref{eqmain3} with trust $\tau=\pi+(1-\pi)\lambda$.
Here $\rho$ is the fixed point of $\rho(\cdot)$, the one experiment
the distortion leaves unchanged. The experiment $\rho(\sigma)$
that the agent contemplates varies with $\sigma$, whereas the $\rho$
recovered from her behavior does not.

A leading example of $\rho(\cdot)$ satisfying \eqref{eqfixed} is
Blackwell garbling. Suppose $S=\{s_{1},s_{2}\}$, and consider an
agent who fears that signals are garbled by some undisclosed flaw
in the process: with probability $k$ a realization of $s_{1}$ is
reported as $s_{2}$, and with probability $k'$ a realization of
$s_{2}$ is reported as $s_{1}$, where $k+k'<1$. The garbled experiment
is 
\(
(1-k-k')\,\sigma+(k+k')\,\rho,
\)
where $\rho$ is the uninformative experiment reporting $s_{1}$ with
probability $k'/(k+k')$, and is left unchanged by the garbling. Her
trust in $\sigma$ exceeds $\pi$, the probability she assigns to
the process being undistorted, since garbling preserves part of $\sigma$.

\paragraph{Example 3: Strategic Sender}

The agent may distrust her information because she suspects that the
sender has personal motivations not to be truthful. To illustrate, consider
a $2\times2$ setting with $\Omega=\{\omega_{1},\omega_{2}\}$ and
$S=\{s_{1},s_{2}\}$, and restrict attention to the interior experiments
in which $s_{1}$ is diagnostic of $\omega_{1}$:\footnote{Our representation theorems assume $|\Omega|\geq4$ only to establish
the sufficiency of our axioms. This restriction may consequently be
ignored in applications, but the two-state setting here can also be
read without loss of generality as a binary partition $\{E,E^{c}\}$
of such an $\Omega$, with $\omega_{1}$ standing for $E$ and $\omega_{2}$
for $E^{c}$, and with attention restricted to experiments whose likelihoods
are constant on each cell.} 
\begin{align}
\Sigma^{1}:=\left\{ \begin{bmatrix}\alpha & 1-\alpha\\
\beta & 1-\beta
\end{bmatrix}:0<\beta<\alpha<1\right\} .\label{eqdefSigma}
\end{align}
Suppose that with probability $\tau$ the agent faces a \emph{truthful}
sender, who runs the experiment $\sigma\in\Sigma^{1}$, and with probability
$1-\tau$ a \emph{strategic} sender, who announces $\sigma$ to mimic
the truthful type but reports signals drawn from some $\rho$ of her
choosing.\footnote{More generally one can assume that the strategic sender adheres with
probability $\lambda$ to her commitment to draw a signal from $\sigma$,
but with probability $1-\lambda$ deviates and instead draws from
some optimally chosen $\rho$. This generates a distortion of the
form~\eqref{eqfixed}.} The strategic sender targets only the receivers who take signals
at face value (that is, trust completely) and seeks to raise their
belief in $\omega_{1}$. Since $s_{1}$ is diagnostic of $\omega_{1}$
under every $\sigma\in\Sigma^{1}$, she reports $s_{1}$ regardless
of the state, so that $\rho=\bigl[\begin{smallmatrix}1 & 0\\
1 & 0
\end{smallmatrix}\bigr]$ for every $\sigma\in\Sigma^{1}$. A Bayesian agent who recognizes that the sender
may be strategic will then interpret signals according to our model.

If the strategic sender maximizes a generic objective over posterior
distributions, the optimal $\rho$ may depend on $\sigma$. It is
nevertheless locally constant under either of two conditions. First,
adjustment frictions may make small changes in $\sigma$ insufficient
to induce a change in $\rho$. Second, if $\rho$ is chosen from a
finite feasible set, then any strictly optimal choice remains optimal
in a neighborhood of $\sigma$. Thus our model applies to any open
$\Sigma$ over which the optimal $\rho$ is constant. Finite feasible
sets arise naturally. For example, a student may choose among courses
on the same topic taught by professors of differing strictness. In
a job interview, he may report only the topic, suggesting the typical
grading standard $\sigma$, while having chosen a course with a lenient
grading standard $\rho$. \citet{FarinaHerman2026} study communication
in environments of this kind.\footnote{In \citet{FarinaHerman2026} the sender observes the state before choosing,
so she may select a different experiment at each $\omega$. Write
$\rho(\cdot\mid\omega)$ for the signal distribution she induces at
$\omega$, namely the row corresponding to $\omega$ in the experiment
she selects there. Collecting these rows gives a single experiment
$\rho$, which generates exactly the signals the receiver observes.
Our model therefore applies.}

\subsection{The Role of \texorpdfstring{$\Sigma$}{Sigma}}

\label{subsec:role-of-Sigma}

Representation~\eqref{eqmain3} holds $\rho$ fixed as the stated
experiment $\sigma$ varies over $\Sigma$. This is a restriction
on the agent's behavior across $\Sigma$, and its plausibility can
depend on $\Sigma$. To illustrate this, recall that in the strategic
sender example above, we saw that $\rho=\bigl[\begin{smallmatrix}1 & 0\\
1 & 0
\end{smallmatrix}\bigr]$ is invariant over $\Sigma^{1}$. This ceases to be so once the space
of experiments is enlarged beyond $\Sigma^{1}$. Let $\Sigma\supsetneq\Sigma^{1}$
be the collection of all interior experiments, which also contains
experiments in which $s_{1}$ is diagnostic of $\omega_{2}$, namely
those with $\alpha<\beta$. Facing such a $\sigma$, the strategic
sender reports $s_{2}$ instead, so that her signals are drawn from
$\rho'=\bigl[\begin{smallmatrix}0 & 1\\
0 & 1
\end{smallmatrix}\bigr]$. No single $\rho$ describes her reports across all of $\Sigma$.
The invariance of $\rho$ is thus a property of the set of experiments,
and one the analyst secures by restricting its range.

Representation~\eqref{eqmain3} also holds $\tau$ fixed as $\sigma$
varies over $\Sigma$. One can imagine a deeper inference process
where the agent infers the trustworthiness of the sender from the
$\sigma$ they offer. Here, ``being told $\sigma$'' is itself
a signal in an inference problem whose states are whether the sender is truthful, and the signals are ``stated experiments''.
Trust would then be a function $\tau(\sigma)$, varying with the announcement
in the way any posterior varies with its signal. The representation
therefore requires a $\Sigma$ over which the sender's credibility
does not vary.

{The property that $\tau$ and $\rho$ are fixed across $\sigma\in\Sigma$
is essentially a consequence of Informational STP (this is evident
in Theorem \ref{thmmain4belief} below). Thus, to satisfy Informational STP when reacting to information from a source that claims experiments in $\Sigma$, one
must commit to $\tau$ and $\rho$ when evaluating any information
it provides. If one changes $\tau$ or $\rho$ with information,
they may violate Informational STP and thereby open themselves up
to exploitation.}

\section{Projective Geometry of Bayesian Updating}

\label{subsec:Proof outline}

The preceding sections characterize Bayesian updating with complete and
limited trust. Although the two representations differ in how the agent
interprets the experiment, in both cases the posterior map preserves collinearity, in a sense that we make precise below. As we show below, this property follows from Informational STP. Using a recent version of
the \emph{Fundamental Theorem of Projective Geometry}, we show that
this geometric property yields a parameterized family
of updating rules that generalize Bayesian updating. Our models of
complete-trust and limited-trust updating are special cases of the family.

Fix a signal realization $s\in S$
and define 
\(
\phi_{s}:\Sigma_{s}\rightarrow\Delta(\Omega),
\)
where $\phi_{s}(\sigma_{s})$ is the posterior belief underlying $\succsim_{\sigma,s}$.
Thus, $\phi_{s}$ maps each likelihood vector 
\(
\sigma_{s}=\bigl(\sigma(s\mid\omega)\bigr)_{\omega\in\Omega}\in\Sigma_{s}
\)
into the posterior induced upon observing $s$. Under complete-trust
Bayesian updating, $\phi_{s}(\sigma_{s})=\left(\frac{p(\omega)\sigma(s\mid\omega)}{\sum_{\omega'\in\Omega}p(\omega')\sigma(s\mid\omega')}\right)_{\omega\in\Omega}.$
We pin down a class of updating rules $\phi_{s}$ through the geometric
structure it preserves, using tools from projective geometry. 

\subsection{Geometric Content of Axiom~\ref{a1}}

Informational STP (Axiom~\ref{a1}) -- a condition stated in terms
of preferences -- implies that the posterior map $\phi_s$
\emph{preserves collinearity}: whenever three likelihood vectors
in $\Sigma_s$ are collinear, their images under $\phi_s$ are
collinear as well.
\begin{defn}
    $\phi_s:\Sigma_s\rightarrow \Delta(\Omega)$ preserves collinearity if, for each pair $\sigma_s,\sigma'_s\in\Sigma_s$ and each $\alpha\in[0,1]$,
there exists $\gamma\in[0,1]$ such that
\(
\phi_s\bigl(\alpha\sigma_s+(1-\alpha)\sigma'_s\bigr)
=\gamma\phi_s(\sigma_s)+(1-\gamma)\phi_s(\sigma'_s).
\)
\end{defn}

Mixture linearity requires the equality to hold with $\gamma=\alpha$ for every $\alpha$, $\sigma_s$, and $\sigma'_s$. Here $\gamma$ may depend on all three. Indeed, because conditional preferences admit SEU representations with a common utility index, Informational STP implies, by a standard separating-hyperplane argument, that the posterior under any mixture of experiments lies in the convex hull of their respective posteriors.
\subsection{Necessity of Axiom \ref{a1}}

Both complete- and limited-trust Bayesian updating preserve
collinearity. This shared geometric property establishes the
necessity of Axiom~\ref{a1} and underlies our axiomatic
characterizations. Figure~\ref{fig:projective-updating2}
illustrates the geometry in these two cases.
Let $\Omega=\{\omega_1,\omega_2,\omega_3\}$ with a uniform prior.
Under full trust, Bayesian updating maps $\sigma_s$ to
$\sigma_s/\|\sigma_s\|_1$, a central projection from the origin
onto the simplex that preserves collinearity. Limited-trust updating adds a \emph{translation} before applying
the same \emph{central projection from the origin}. To see this,
suppose the agent attaches trust $\tau\in(0,1]$ to $\sigma$
and $1-\tau$ to $\rho$. The posterior is then the normalization
of $\tau\sigma_s+(1-\tau)\rho_s$. Since normalization is
unaffected by positive scaling, the same posterior is obtained
by normalizing $\sigma_s+\frac{1-\tau}{\tau}\rho_s$.
The updating rule can therefore be written as:
\[
\sigma_s
\;\xrightarrow{\ \text{translation}\ }\;
\sigma_s+\tfrac{1-\tau}{\tau}\rho_s
\;\xrightarrow{\ \text{projection}\ }\;
\left(\sigma_s+\frac{1-\tau}{\tau}\rho_s\right)
\Big/\left\|\sigma_s+\frac{1-\tau}{\tau}\rho_s\right\|_1.
\]
Since both transformations preserve collinearity, so does
limited-trust updating.
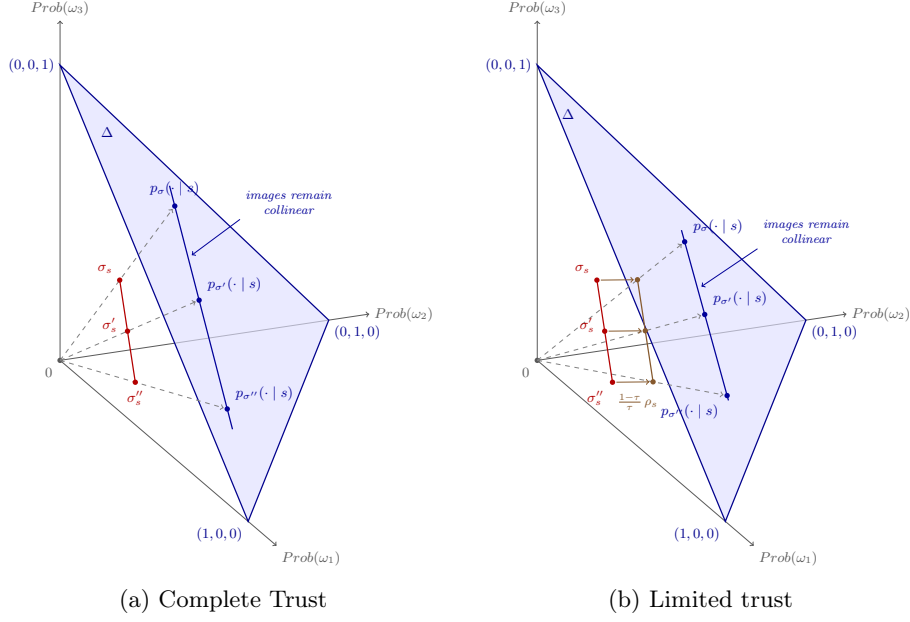
\begin{figure}[htbp]
\centering
 
\begin{subfigure}[t]{0.48\textwidth}
\centering
\resizebox{\linewidth}{!}{%
\begin{tikzpicture}[trust diagram]
 
% ================= coordinates =================
\coordinate (O)  at (0,0,0);
\coordinate (E1) at (1,0,0);
\coordinate (E2) at (0,1,0);
\coordinate (E3) at (0,0,1);
 
% interior collinear points, deep inside the tetrahedron
\coordinate (S1) at (0.30,0.07,0.08);
\coordinate (S2) at (0.18,0.125,0.18);
\coordinate (S3) at (0.06,0.18,0.28);
 
% central projections onto the simplex
\coordinate (P1) at (0.666667,0.155556,0.177778);
\coordinate (P2) at (0.371134,0.257732,0.371134);
\coordinate (P3) at (0.115385,0.346154,0.538462);
 
% slight extension of the image line
\coordinate (P0) at (0.721795,0.136496,0.141709);
\coordinate (P4) at (0.060257,0.365214,0.574530);
 
% ================= axes =================
\draw[->,gray!70!black] (O) -- (1.10,0,0)  node[above right]        {$Prob(\omega_1)$};
\draw[->,gray!70!black] (O) -- (0,1.225,0) node[above left]         {$Prob(\omega_2)$};
\draw[->,gray!70!black] (O) -- (0,0,1.18)  node[right,inner sep=2pt] {$Prob(\omega_3)$};
\node[dot,gray!70!black,label={[gray!70!black]below left:$0$}] at (O) {};
 
% ================= probability simplex =================
\fill[blue!12,opacity=.65] (E1) -- (E2) -- (E3) -- cycle;
\draw[blue!55!black,thick]  (E1) -- (E2) -- (E3) -- cycle;
\node[blue!55!black,left,inner sep=2pt,xshift=-2pt,yshift=-5pt] at (E1) {$(1,0,0)$};
\node[blue!55!black,below right]                               at (E2) {$(0,1,0)$};
\node[blue!55!black,above right,inner sep=2pt]                 at (E3) {$(0,0,1)$};
\node[blue!55!black] at (0.08,0.12,0.80) {$\Delta$};
 
% ================= dashed projection rays =================
\foreach \p in {P1,P2,P3}
  \draw[dashed,gray,->,shorten >=2.5pt] (O) -- (\p);
 
% ================= interior collinear points =================
\draw[red!70!black,thick] (S1) -- (S3);
\node[dot,red!70!black] at (S1) {};
\node[dot,red!70!black] at (S2) {};
\node[dot,red!70!black] at (S3) {};
\node[red!70!black,below right=1pt] at (S1) {$\sigma''_s$};
\node[red!70!black,above right=1pt] at (S2) {$\sigma'_s$};
\node[red!70!black,above left=1pt]  at (S3) {$\sigma_s$};
 
% ================= image line on the simplex =================
\draw[blue!60!black,thick] (P0) -- (P4);
\node[dot,blue!60!black] at (P1) {};
\node[dot,blue!60!black] at (P2) {};
\node[dot,blue!60!black] at (P3) {};
\node[blue!60!black,above right=2pt]            at (P1) {$p_{\sigma''}(\cdot\mid s)$};
\node[blue!60!black,above right=2pt]            at (P2) {$p_{\sigma'}(\cdot\mid s)$};
\node[blue!60!black,above left=2pt,xshift=-4pt] at (P3) {$p_{\sigma}(\cdot\mid s)$};
 
% ================= annotation =================
\node[blue!60!black,font=\footnotesize\itshape,align=center] (note)
      at (0,0.74,0.54) {images remain\\ collinear};
\draw[blue!60!black,->,shorten >=3pt] (note.south west) -- ($(P2)!0.45!(P3)$);
 
\end{tikzpicture}
}
\caption{Complete Trust}
\label{fig:bayes-projective2}
\end{subfigure}
\hfill
\begin{subfigure}[t]{0.48\textwidth}
\centering
\resizebox{\linewidth}{!}{%
\begin{tikzpicture}[trust diagram]
 
% ================= coordinates =================
\coordinate (O)  at (0,0,0);
\coordinate (E1) at (1,0,0);
\coordinate (E2) at (0,1,0);
\coordinate (E3) at (0,0,1);
 
% same interior collinear points as Figure 1
\coordinate (S1) at (0.30,0.07,0.08);
\coordinate (S2) at (0.18,0.125,0.18);
\coordinate (S3) at (0.06,0.18,0.28);
 
% translated points: shift by v = (1-tau)/tau * rho_s = (0.03,0.13,0)
\coordinate (T1) at (0.33,0.20,0.08);
\coordinate (T2) at (0.21,0.255,0.18);
\coordinate (T3) at (0.09,0.31,0.28);
 
% central projections of the translated points onto the simplex
\coordinate (Q1) at (0.540984,0.327869,0.131148);
\coordinate (Q2) at (0.325581,0.395349,0.279070);
\coordinate (Q3) at (0.132353,0.455882,0.411765);
 
% slight extension of the image line
\coordinate (Q0) at (0.553908,0.323820,0.122272);
\coordinate (Q4) at (0.101436,0.465568,0.432996);
 
% ================= axes =================
\draw[->,gray!70!black] (O) -- (1.10,0,0)  node[above right]        {$Prob(\omega_1)$};
\draw[->,gray!70!black] (O) -- (0,1.225,0) node[above left]         {$Prob(\omega_2)$};
\draw[->,gray!70!black] (O) -- (0,0,1.18)  node[right,inner sep=2pt] {$Prob(\omega_3)$};
\node[dot,gray!70!black,label={[gray!70!black]below left:$0$}] at (O) {};
 
% ================= probability simplex =================
\fill[blue!12,opacity=.65] (E1) -- (E2) -- (E3) -- cycle;
\draw[blue!55!black,thick]  (E1) -- (E2) -- (E3) -- cycle;
\node[blue!55!black,left,inner sep=2pt,xshift=-2pt,yshift=-5pt] at (E1) {$(1,0,0)$};
\node[blue!55!black,below right]                               at (E2) {$(0,1,0)$};
\node[blue!55!black,above right,inner sep=2pt]                 at (E3) {$(0,0,1)$};
\node[blue!55!black] at (0.08,0.12,0.80) {$\Delta$};
 
% ================= dashed projection rays =================
\foreach \q in {Q1,Q2,Q3}
  \draw[dashed,gray,->,shorten >=2.5pt] (O) -- (\q);
 
% ================= original collinear points =================
\draw[red!70!black,thick] (S1) -- (S3);
\node[dot,red!70!black] at (S1) {};
\node[dot,red!70!black] at (S2) {};
\node[dot,red!70!black] at (S3) {};
\node[red!70!black,below=2pt,xshift=1.5pt,inner sep=1pt]      at (S1) {$\sigma''_s$};
\node[red!70!black,above right=1pt,xshift=2.5pt,yshift=-1pt] at (S2) {$\sigma'_s$};
\node[red!70!black,above left=1pt]                            at (S3) {$\sigma_s$};
 
% ================= translation arrows =================
\draw[brown!70!black,->,shorten <=2.5pt,shorten >=2.5pt] (S1) -- (T1)
      node[midway,below=2pt,xshift=10pt,font=\footnotesize] {$\frac{1-\tau}{\tau}\,\rho_s$};
\draw[brown!70!black,->,shorten <=2.5pt,shorten >=2.5pt] (S2) -- (T2);
\draw[brown!70!black,->,shorten <=2.5pt,shorten >=2.5pt] (S3) -- (T3);
 
% translated points
\draw[brown!70!black,thick] (T1) -- (T3);
\node[dot,brown!70!black] at (T1) {};
\node[dot,brown!70!black] at (T2) {};
\node[dot,brown!70!black] at (T3) {};
 
% ================= image line on the simplex =================
\draw[blue!60!black,thick] (Q0) -- (Q4);
\node[dot,blue!60!black] at (Q1) {};
\node[dot,blue!60!black] at (Q2) {};
\node[dot,blue!60!black] at (Q3) {};
\node[blue!60!black,below=3pt]                  at (Q1) {$p_{\sigma''}(\cdot\mid s)$};
\node[blue!60!black,above right=2pt]            at (Q2) {$p_{\sigma'}(\cdot\mid s)$};
\node[blue!60!black,above left=2pt,xshift=-4pt] at (Q3) {$p_{\sigma}(\cdot\mid s)$};
 
% ================= annotation =================
\node[blue!60!black,font=\footnotesize\itshape,align=center] (note)
      at (0,0.84,0.43) {images remain\\ collinear};
\draw[blue!60!black,->,shorten >=3pt] (note.south west) -- ($(Q2)!0.55!(Q3)$);
 
\end{tikzpicture}
}
\caption{Limited trust}
\label{fig:limited-trust-projective2}
\end{subfigure}
 
\caption{Preservation of Collinearity}
\label{fig:projective-updating2}
\end{figure}

\subsection{Sufficiency of Axiom \ref{a1}}

To establish sufficiency, we invoke the Fundamental Theorem
of Projective Geometry, reviewed in Appendix \ref{app:b}. The classical versions of the theorem
(\citealp{FaureFroelicher1994,Havlicek1994}, see also Theorem~1.18 in \citealp{SanchoDeSalas2026}) apply to mappings with an unbounded
domain, whereas the domain of an updating rule---a set of likelihood
vectors---is bounded. The recent extension by \citet[Theorem~3.11]{SanchoDeSalas2026} covers bounded domains and enables us to establish the
following key mathematical result, stated as Theorem~\ref{thmrepresentation2}
in the appendix: \emph{for $n\geq4$ and open, convex $C\subset\mathbb{R}^{n}$,
if a function $f:C\rightarrow\mathbb{R}^{n}$ preserves collinearity
and its image is not contained in a plane, then there exist $c\in\mathbb{R}$,
$a,b\in\mathbb{R}^{n}$, and $M\in\mathbb{R}^{n\times n}$ such that
} 
\begin{align}
f(v)=\frac{M\cdot v+a}{b^{\top}\cdot v+c},\qquad v\in C.\label{eq}
\end{align}
In turn, specializing this abstract result to updating rules $\phi_{s}:\Sigma_{s}\rightarrow\Delta(\Omega)$
(which have very particular domain and range spaces) leads us to establish
that:

\begin{thm}\label{thmmain4belief} Suppose the image of $\phi_{s}:\Sigma_{s}\rightarrow\Delta(\Omega)$
is not contained in a plane in $\Delta(\Omega)$. Then the following
statements are equivalent: 
\begin{itemize}
\item $\phi_{s}$ preserves collinearity;
\item there exist a quasi-experiment $\rho$, $\tau\in(0,1]$, and a $|\Omega|\times|\Omega|$
matrix $M_{s}$ such that 
\begin{align}
\phi_{s}(\sigma_{s})=\left(\frac{[\tau M_{s}\cdot\sigma_{s}+(1-\tau)\rho_{s}]_{\omega}\cdot p(\omega)}{\sum_{\omega'}[\tau M_{s}\cdot\sigma_{s}+(1-\tau)\rho_{s}]_{\omega'}\cdot p(\omega')}\right)_{\omega\in\Omega}.\label{eqthmmain4}
\end{align}
\end{itemize}
\end{thm}

While Theorem~\ref{thmmain4belief} is belief-theoretic, its decision-theoretic
counterpart (stated in terms of $\succsim_{\sigma,s}$ rather
than $\phi_{s}$) is given as Theorem~\ref{thmmain4} in the appendix,
and is our most general representation result. Theorems~\ref{thmmain1} and~\ref{thmmain2} are specializations
of Theorem~\ref{thmmain4}.

The representation (\ref{eqthmmain4}) reads as Bayesian updating
with prior $p$ applied to the \emph{distorted} likelihood vector
\(
\tau M_{s}\cdot\sigma_{s}+(1-\tau)\rho_{s}.
\)
Each component has a natural interpretation: $M_{s}\cdot\sigma_{s}$
is the ``likelihood vector'' in $\sigma$ as perceived by the agent,
$\rho_{s}$ is her alternative interpretation of the signal realization
in case $\sigma$ is not the true signal structure, and $\tau$ is
her confidence in $\sigma$. The matrix $M_{s}$ thus permits systematic
distortions in how the agent perceives likelihoods---distortions
that are not attributable to lack of trust. In our representation
results, Axioms~\ref{a2} and~\ref{a4} force $M_{s}$ to be the
identity matrix, thereby ruling out such perceptual distortions and
reducing the family to Bayesian updating with limited trust.

\section{Systematic Non-Homogeneity}

Agents often react to surprising information differently from what Bayes' rule with complete trust prescribes: they may dismiss
a surprise as noise, or take it more seriously than it warrants \citep[see][for the evidence]{Ortoleva2012,Ortoleva2024}. Such responses
arise naturally when Homogeneity fails. Scaling a likelihood vector
changes the ex ante probability of the signal without changing its
likelihood ratios. Thus, when Homogeneity is violated, the agent's
posterior response may depend on how surprising the signal is, even
when its informational content about the state is held fixed. This
section defines underreaction and overreaction to surprises behaviorally
and characterizes them via the alternative theory $\rho$.

\subsection{Underreaction to Surprises}

We begin by providing a behavioral definition.

\begin{ax}[Underreaction to Surprises]\label{a:under} For $\sigma'_{s}=\alpha\sigma{}_{s}$, $\alpha\in(0,1)$, if 
$f\succsim_{0}g\text{ and }f\succsim_{\sigma,s}g$, then $f\succsim_{\sigma',s}g$.
\end{ax}

Axiom~\ref{a:under} is a weakening of \textit{Homogeneity} that
restricts its violations to underreaction to surprises. Equivalently, if $\sigma'_{s}=\alpha\sigma{}_{s}$
for $\alpha\in(0,1)$, then 
\(
f\succsim_{\sigma,s}g\text{ and }f\prec_{\sigma',s}g\implies f\prec_{0}g.
\)
That is, whenever we observe the Homogeneity violation given by $f$
weakly chosen over $g$ given $\sigma$ but $g$ strictly chosen over
$f$ given $\sigma'$, it must be because $g$ is strictly preferred
over $f$ according to the \textit{ex-ante} preference. In this sense,
making the signal more unexpected can only move the agent's ranking
toward her ex ante ranking, never away from it. The axiom thus concerns
how the agent responds when the same information arrives more unexpectedly,
rather than prescribing conservatism at every signal, as in \citet*{EpsteinNoorSandroni2008} and \citet{Kovach2021}. \citet[Figure~12]{saeed2026decomposing} finds that, holding the likelihood ratio fixed, subjects updated less on average when the observed signal was less likely under both states, a pattern that suggests the preferences in Axiom~\ref{a:under}.

% Axiom \ref{a:under} has the same form as Consistency with Prior (Axiom~\ref{a4}),
% but applies to a multiplicative rather than an additive perturbation
% of the likelihood vector. Axiom~\ref{a4} makes the signal less informative,
% whereas Axiom~\ref{a:under} makes it less likely without changing
% its likelihood ratios. %In both cases, the agent's ranking moves toward
% her ex ante ranking. Thus, an agent satisfying Axiom~\ref{a:under}
% reacts less strongly when the same information is more unexpected,
% treating part of the surprise as noise rather than as evidence about
% the state.

\begin{prop}\label{prop:under} Suppose $\{\succsim_{0},\succsim_{\sigma,s}\}_{(\sigma,s)\in\Sigma\times S}$
satisfies~\eqref{eqmain3} with $\tau<1$. Then Axiom~\ref{a:under}
holds if and only if  
\(
\rho(s\mid\omega)=\rho(s\mid\omega')
\)
for all $s\in S\text{ and }\omega,\omega'\in\Omega$, \textit{i.e.} $\rho$ is uninformative.
\end{prop}

Underreaction to surprises therefore corresponds to an entirely uninformative
alternative theory. As a signal becomes less likely, the agent assigns
greater relative importance to a process that reveals nothing about
the state, pulling her posterior toward her prior.

Axiom~\ref{a:under} is a comparative static on how the agent responds
when a signal becomes less likely. The next result relates it to a
global notion of underreaction. We say that the agent's beliefs exhibit
\emph{conservatism} if her posterior underreacts relative to the Bayesian
update for all experiments $\sigma\in\Sigma$ and signals $s\in S$, i.e. $p^{*}_{\sigma}(\omega|s)\le p_{\sigma}(\omega|s)\le p(\omega)$ or $p(\omega)\le p_{\sigma}(\omega|s)\le p^{*}_{\sigma}(\omega|s)$ for all $\omega\in\Omega$. While violations have been documented in the literature, conservatism
is a notably robust finding in experiments \citep{Benjamin2019}. We observe
that

\begin{prop}\label{prop-conservatism} Suppose $\{\succsim_{0},\succsim_{\sigma,s}\}_{(\sigma,s)\in\Sigma\times S}$
satisfies~\eqref{eqmain3}. If Underreaction to Surprises holds then
the agent's beliefs exhibit conservatism. The converse holds if $\Sigma_{s}$
contains at least one constant likelihood vector, for each $s$. \end{prop}

By Proposition \ref{prop:under}, an uninformative
$\rho$ generates conservatism in beliefs. The converse does not necessarily
hold. Intuitively, even if $\rho$ is informative, the agent will
exhibit conservatism as long as it is \emph{less} informative than
all $\sigma\in\Sigma$. The proposition shows that the converse holds
nevertheless under a very mild richness condition on $\Sigma$.

\subsection{Overreaction to Surprises}

The characterization of overreaction is not simply the converse
of that for underreaction. We first
present the overreaction analog of Axiom \ref{a:under}.

\begin{ax}[Overreaction to Surprises] \label{a:over} For $\sigma'_{s}=\alpha\sigma{}_{s}$, $\alpha\in(0,1)$, if $f\succsim_{0}g\text{ and }f\succsim_{\sigma',s}g$, then $f\succsim_{\sigma,s}g$.
\end{ax}

While Axiom \ref{a:over} (like Axiom \ref{a:under}) is ``global'' in that it applies to all experiments, one can also define a ``local'' version: for a fixed experiment $\sigma\in\Sigma$, the agent
\emph{overreacts to surprises at $\sigma$} if the implication in
Axiom~\ref{a:over} holds for every $s\in S$ and every
$\sigma'\in\Sigma$ such that $\sigma'_s=\alpha\sigma_s$ for some
$\alpha\in(0,1)$.

% To illustrate why (\ref{eq:over}) restricts
% violations of Homogeneity at $\sigma$ to overreaction to surprises,
% consider an equivalent formulation of it: for all $s\in S$ and all
% $\sigma'\in\Sigma$ with $\sigma'_{s}=\alpha\sigma{}_{s}$ for some
% $\alpha\in(0,1)$, 
% \(
% f\succsim_{0}g\text{ and }g\succ_{\sigma,s}f\implies g\succ_{\sigma',s}f.
% \)
% That is, whenever the agent's ranking departs from her \textit{ex-ante}
% ranking at $\sigma$, it must also depart at $\sigma'$, which scales
% down the likelihood of signal $s$. In this sense, making the signal
% more unexpected can only move the agent's ranking away from her ex
% ante ranking, never back toward it.

One could imagine a ``pathological'' reaction to signals whereby
the posterior moves in the opposite direction to the Bayesian update,
an effect known as the \emph{backfire effect}.\footnote{See \citet{NyhanReifler2010,NyhanReifler2015} and \citet{Haglin2017}.}
While possible in our model,\footnote{A backfire effect occurs in the model when $\sigma$ and $\pi:=\tau\sigma+(1-\tau)\rho$
push beliefs in different directions at some signal. For instance,
if $\sigma(s|\omega)>\sigma_{p}(s)$ and $\pi(s|\omega)\le\pi_{p}(s)$
then $\sigma$ suggests a higher belief in $\omega$ after signal
$s$ but the agent does not raise her belief in $\omega$ because
she follows $\pi$.} we exclude it in our results below for a more focused analysis. Say
that the agent exhibits \emph{no backfire} at $\sigma$ if for all
$s\in S$ and $\omega\in\Omega$, 
\(
p^{*}_{\sigma}(\omega|s)\leq p(\omega)\iff p_{\sigma}(\omega|s)\leq p(\omega).
\)
Under this assumption we obtain the following characterization of
overreaction to surprises. \begin{prop}\label{prop: overreaction0}
Suppose $\{\succsim_{0},\succsim_{\sigma,s}\}_{(\sigma,s)\in\Sigma\times S}$
satisfies~\eqref{eqmain3} with $\tau<1$. Then: 
\begin{itemize}
\item Axiom \ref{a:over} cannot hold; 
\item If the agent exhibits no backfire at $\sigma$, then she overreacts
to surprises at $\sigma$ if and only if for each $s\in S$ there
are $\alpha_{s},\beta_{s}\ge0$ such that 
\(
\rho_{s}=\alpha_{s}\sigma_{s}-\beta_{s}1_{\Omega}.
\)
\end{itemize}
\end{prop}

The first part of Proposition~\ref{prop: overreaction0} states that
overreaction to surprises cannot hold at every experiment in $\Sigma$.
This is a sharp asymmetry with conservatism: an uninformative alternative
theory generates underreaction to surprises at every experiment, whereas
no $\rho$ generates overreaction everywhere, for any $\Sigma$ satisfying
our assumptions. The second part gives a local characterization. At
a given experiment, overreaction requires the alternative likelihood
vector to be an affine transformation of the agent's likelihood vector,
with a negative constant component.

The last result confirms that overreaction to surprises implies overreaction
to information. We say that the agent exhibits\emph{ overreaction}
at $\sigma\in\Sigma$, if for all $s\in S$ and $\omega\in\Omega$, $p_{\sigma}(\omega|s)\leq p^{*}_{\sigma}(\omega|s)\le p(\omega)$ or $p(\omega)\le p^{*}_{\sigma}(\omega|s)\leq p_{\sigma}(\omega|s)$.

\begin{prop}\label{prop: overreaction} Suppose $\{\succsim_{0},\succsim_{\sigma,s}\}_{(\sigma,s)\in\Sigma\times S}$
satisfies~\eqref{eqmain3}. Then if the agent exhibits no backfire
and overreacts to surprises at $\sigma$, then the agent exhibits
overreaction at $\sigma$.\end{prop}

\section{Prior Dependence}

\label{sec:Prior-Dependence}

The preceding analysis explores how posteriors vary with $\sigma\in\Sigma$
for a fixed prior $p$. This section explores how posteriors
vary with the prior $p$ at a fixed $\sigma\in\Sigma$. Such dependence
is, for instance, the content of \emph{confirmation bias}, the well-known
finding that the response to information is biased towards the prior,
which we emphasize below.\footnote{\citet{GentzkowShapiro2006} observe that posterior trust in an information
source depends on the prior and suggest a connection between lack
of trust and confirmation bias without clarifying the behavioral content
of the connection. The results of this section can be viewed as filling
this gap. } For simplicity, we restrict attention throughout to
the $2\times2$ setting of Section \ref{sec32}, where $\Omega=\{\omega_{1},\omega_{2}\}$, {$S=\{s_{1},s_{2}\}$},
and $\sigma$ lies in the set $\Sigma^{1}$ defined in (\ref{eqdefSigma}),
that is, the (interior of the) set of experiments in which $s_{1}$
is diagnostic of $\omega_{1}$. We identify $p\in(0,1)$ with the
prior $(p,1-p)$.

\subsection{Non-Martingale Properties}

We assume $\sigma$ is an objectively true experiment,
determining the state-contingent signal frequencies, but
the agents do not fully trust it. We consider a family
of agents with the same objective experiment $\sigma$, subjective
experiment $\rho$, and trust parameter $\tau$, but different
priors $p$, denoted $(\sigma,\rho,\tau,p)_{p\in(0,1)}.$ We define
behavioral biases via systematic violations of
the martingale property of Bayesian updating. For a prior $p$ and
a state $\omega$, write 
\[
\bar{q}_{p}(\omega):=\sum_{s\in S}p^{*}_{\tau\sigma+(1-\tau)\rho}(\omega\mid s)\,\sigma_{p}(s)
\]
for the agent's average posterior on $\omega$, where the posterior
is formed under the subjective experiment $\tau\sigma+(1-\tau)\rho$
while signals arrive according to the objective experiment $\sigma$.
A fully trusting Bayesian has $\bar{q}_{p}(\omega)=p(\omega)$.

\begin{defn}[Behavioral Biases]\label{Def: Conf Updating}
A family of agents $(\sigma,\rho,\tau,p)_{p\in(0,1)}$ exhibits 
\begin{itemize}
\item \emph{$\omega$-biased updating}, for $\omega\in\{\omega_{1},\omega_{2}\}$,
if $\bar{q}_{p}(\omega)>p(\omega)$ for all $p\in(0,1)$; 
\item \emph{confirmatory updating} if there exists $\theta\in(0,1)$ such
that, for any $\omega\in\Omega$, $\bar{q}_{p}(\omega)>p(\omega)$
whenever $p(\omega)>\theta$; 
\item \emph{average base-rate neglect} if there exists $\theta\in(0,1)$
such that, for any $\omega\in\Omega$, $\bar{q}_{p}(\omega)>p(\omega)$
whenever $p(\omega)<\theta$. 
\end{itemize}
\end{defn}

On average, beliefs move toward state $\omega$ under $\omega$-biased updating, reinforce an extreme prior under confirmatory updating, and temper it
under average base-rate neglect. The next result shows that these three biases are \emph{exhaustive}:
in our model, any violation of the martingale property takes one of
these forms.

\begin{prop}\label{propconfirmation} Let $(\sigma,\rho,\tau,p)_{p\in(0,1)}$
be such that $\tau<1$ and there is no backfire at $\sigma$. The
family of agents
\begin{enumerate}[label=(\alph*),ref=\alph*,nosep]
\item\label{prop:bias-state} exhibits $\omega_{1}$-biased updating if and only if
${\sigma(s_{1}\mid\omega_{1})}\geq{\rho(s_{1}\mid\omega_{1})}$ and ${\sigma(s_{1}\mid\omega_{2})}\geq{\rho(s_{1}\mid\omega_{2})}$,
where at least one of the two inequalities is strict;
\item\label{prop:bias-confirmatory} exhibits confirmatory updating if and only if
${\sigma(s_{1}\mid\omega_{1})}>{\rho(s_{1}\mid\omega_{1})}$ and ${\sigma(s_{2}\mid\omega_{2})}>{\rho(s_{2}\mid\omega_{2})}$;
\item\label{prop:bias-neglect} exhibits average base-rate neglect if and only if
${\sigma(s_{1}\mid\omega_{1})}<{\rho(s_{1}\mid\omega_{1})}$ and ${\sigma(s_{2}\mid\omega_{2})}<{\rho(s_{2}\mid\omega_{2})}$.
\end{enumerate}
\end{prop}

Proposition~\ref{propconfirmation} characterizes the regions of
the experiment domain in which different biases arise. Under its assumptions,
the three types of bias exhaust all violations of the martingale property.
The type of bias depends only on the position of $\sigma$ relative
to $\rho$: the degree of trust $\tau$ affects the magnitude of the
bias, but not its form.

Panel~(a) of Figure~\ref{fig:biases} illustrates the regions characterized
by Proposition~\ref{propconfirmation}(\ref{prop:bias-state})--(\ref{prop:bias-neglect}). Each experiment $\sigma$ is
represented by the point $\bigl(\sigma(s_{1}\mid\omega_{1}),\sigma(s_{2}\mid\omega_{2})\bigr)$.
We focus on the region above the anti-diagonal, where $s_{1}$ is
diagnostic of $\omega_{1}$. The condition in Proposition~\ref{propconfirmation}(\ref{prop:bias-state}) defines the
region weakly southeast of $\rho$, where updating is $\omega_{1}$-biased;
the northwest region gives the corresponding $\omega_{2}$-bias. The conditions in Proposition~\ref{propconfirmation}(\ref{prop:bias-confirmatory}) and~(\ref{prop:bias-neglect}) define the northeast and southwest regions, corresponding
to confirmatory updating and average base-rate neglect, respectively.

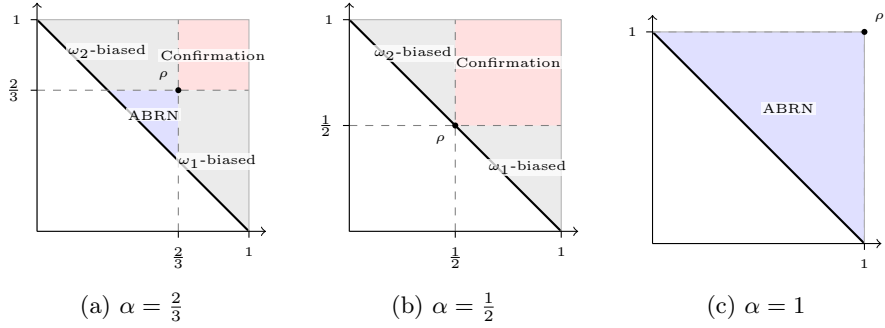
\begin{figure}[htbp]
\centering
% ---------- (a) alpha = 2/3 ----------
\begin{subfigure}[t]{0.32\textwidth}
\centering
\begin{tikzpicture}[bias diagram]
\fill[red!12]   (0.6667,0.6667) -- (1,0.6667) -- (1,1) -- (0.6667,1) -- cycle;
\fill[blue!12]  (0.3333,0.6667) -- (0.6667,0.6667) -- (0.6667,0.3333) -- cycle;
\fill[black!8]  (0.6667,0.3333) -- (1,0) -- (1,0.6667) -- (0.6667,0.6667) -- cycle;
\fill[black!8]  (0.3333,0.6667) -- (0.6667,0.6667) -- (0.6667,1) -- (0,1) -- cycle;
\draw[thick] (0,1) -- (1,0);
\draw[dashed,gray] (0.6667,0) -- (0.6667,1);
\draw[dashed,gray] (0,0.6667) -- (1,0.6667);
\draw[gray!60] (0,0) rectangle (1,1);
\draw[->] (0,0) -- (1.08,0);
\draw[->] (0,0) -- (0,1.08);
\foreach \x/\l in {0.6667/{\tfrac{2}{3}}, 1/1}{
  \draw (\x,0) -- (\x,-0.03) node[below] {$\l$};
  \draw (0,\x) -- (-0.03,\x) node[left]  {$\l$};}
\node[dot] at (0.6667,0.6667) {};
\node[above left=0.3pt] at (0.6667,0.6667) {$\rho$};
\node[lbl] at (0.83,0.83) {Confirmation};
\node[lbl] at (0.55,0.55) {ABRN};
\node[lbl] at (0.85,0.33) {$\omega_{1}$-biased};
\node[lbl] at (0.33,0.85) {$\omega_{2}$-biased};
\end{tikzpicture}
\caption{$\alpha=\tfrac{2}{3}$}
\label{fig:bias-23}
\end{subfigure}
\hfill
% ---------- (b) alpha = 1/2 ----------
\begin{subfigure}[t]{0.32\textwidth}
\centering
\begin{tikzpicture}[bias diagram]
\fill[red!12]  (0.5,0.5) -- (1,0.5) -- (1,1) -- (0.5,1) -- cycle;
\fill[black!8] (0.5,0.5) -- (1,0) -- (1,0.5) -- cycle;
\fill[black!8] (0.5,0.5) -- (0,1) -- (0.5,1) -- cycle;
\draw[thick] (0,1) -- (1,0);
\draw[dashed,gray] (0.5,0) -- (0.5,1);
\draw[dashed,gray] (0,0.5) -- (1,0.5);
\draw[gray!60] (0,0) rectangle (1,1);
\draw[->] (0,0) -- (1.08,0);
\draw[->] (0,0) -- (0,1.08);
\foreach \x/\l in {0.5/{\tfrac{1}{2}}, 1/1}{
  \draw (\x,0) -- (\x,-0.03) node[below] {$\l$};
  \draw (0,\x) -- (-0.03,\x) node[left]  {$\l$};}
\node[dot] at (0.5,0.5) {};
\node[below left=0.3pt] at (0.5,0.5) {$\rho$};
\node[lbl] at (0.75,0.80) {Confirmation};
\node[lbl] at (0.84,0.30) {$\omega_{1}$-biased};
\node[lbl] at (0.30,0.84) {$\omega_{2}$-biased};
\end{tikzpicture}
\caption{$\alpha=\tfrac{1}{2}$}
\label{fig:bias-half}
\end{subfigure}
\hfill
% ---------- (c) alpha = 1 ----------
\begin{subfigure}[t]{0.32\textwidth}
\centering
\begin{tikzpicture}[bias diagram]
\fill[blue!12] (0,1) -- (1,0) -- (1,1) -- cycle;
\draw[thick] (0,1) -- (1,0);
\draw[dashed,gray] (1,0) -- (1,1);
\draw[dashed,gray] (0,1) -- (1,1);
\draw[gray!60] (0,0) rectangle (1,1);
\draw[->] (0,0) -- (1.08,0);
\draw[->] (0,0) -- (0,1.08);
\draw (1,0) -- (1,-0.03) node[below] {$1$};
\draw (0,1) -- (-0.03,1) node[left]  {$1$};
\node[dot] at (1,1) {};
\node[above right=0.3pt] at (1,1) {$\rho$};
\node[lbl] at (0.64,0.64) {ABRN};
\end{tikzpicture}
\caption{$\alpha=1$}
\label{fig:bias-one}
\end{subfigure}

\caption{Behavioral biases as regions in
$\bigl(\sigma(s_{1}\mid\omega_{1}),\sigma(s_{2}\mid\omega_{2})\bigr)$.}
\label{fig:biases}

\vspace{0.2cm}

\begin{minipage}[c]{0.9\textwidth}
\footnotesize
\emph{Note:} The horizontal axis is $\sigma(s_{1}\mid\omega_{1})$ and
the vertical axis is $\sigma(s_{2}\mid\omega_{2})$. The point $\rho$
marks
$\bigl(\rho(s_{1}\mid\omega_{1}),\rho(s_{2}\mid\omega_{2})\bigr)$,
where $\alpha:=\rho(s_{1}\mid\omega_{1})=\rho(s_{2}\mid\omega_{2})$
is the precision of the alternative theory. ``ABRN'' denotes average
base-rate neglect.
\end{minipage}
\end{figure}

To interpret these regions, suppose an advisor communicates
$\sigma$ as the DGP of her feedback to a student. The states $\omega_{1}$
and $\omega_{2}$ indicate high- and mediocre-quality research, while
$s_{1}$ and $s_{2}$ are favorable and unfavorable feedback. The
stated DGP is correct, but the student only partially trusts it and
also entertains $\rho$. In the grey southeast region, favorable feedback
is more common than under $\rho$. By placing weight on $\rho$, the
student underestimates the advisor's leniency. Favorable feedback
thus arrives more often than she expects, tilting her beliefs toward
high-quality research on average.

The red northeast region illustrates the confirmation mechanism. Here
$\sigma$ is more accurate than $\rho$: the student doubts the advisor's
stated accuracy and entertains an alternative theory under which feedback
matches research quality less often. If she is sufficiently confident
that her research is strong, unfavorable feedback is more likely under the alternative theory 
$\rho$ than under $\sigma$. Such disconfirming feedback
therefore lowers her posterior confidence in the stated DGP, leading
her to discount the evidence against her prior. The same reasoning
applies when she is sufficiently confident that her research is weak. Average base-rate neglect
is the opposite pattern.

Under the no-backfire condition, this correspondence between quadrants
and biases is universal: whatever $\rho$, the four quadrants it induces
map to the same four patterns. Panels~(b) and~(c) show two special
cases. With uninformative $\rho$, limited trust produces underreaction
to information (Proposition~\ref{prop:under}), and panel~(b) permits
only confirmatory updating or a bias toward one state. With perfectly
informative $\rho$, limited trust produces overreaction to information:
every interior experiment above the anti-diagonal lies southwest of
$\rho$, so panel~(c) permits only average base-rate neglect.

\subsection{ Confirmatory Reading of Mixed Evidence}

\label{subsec:Conf Reading}

The striking finding in the previous subsection is that a prior-dependent
notion such as confirmatory updating arises in the model for a family
of agents for which $\rho$ is prior-independent. While confirmatory
updating is the clearest expression of ``confirmation
bias'' in our model, the notion is defined differently
in the psychology literature. In a classic experiment, \citet*{LordRossLepper1979} demonstrate that offering different subjects the
same evidence on an issue (the effectiveness of capital punishment
in deterring crime) leads them to believe in their prior positions
on the issue more strongly, leading to increased polarization. A key
feature of this study, which is not captured in Definition \ref{Def: Conf Updating},
is that the evidence presented to subjects was \emph{mixed}. Specifically,
subjects were presented with two studies, one that presented empirical
evidence on the issue and one that criticized the first study's methodology. We study a definition of confirmation bias that formalizes this.

Consider a family of agents as before, but now allow $\rho$ to depend
on the prior. Denote the family by $(\sigma,\rho^{p},\tau,p)_{p\in(0,1)}$.
In the $2\times2$ example, each agent now receives the results
of \emph{two} studies. They are told
that both results are essentially signals drawn from one experiment
$\sigma\in\Sigma^{1}$. The agent, however, places only probability
$\tau$ on a given signal being drawn from $\sigma,$ placing the
remaining probability $1-\tau$ on it being drawn from $\rho^{p}$.
Thus, a pair of signals $s,s'\in S$ is viewed as independently generated
by the subjective experiment $\tau\sigma+(1-\tau)\rho^{p}$.

A pair of signals $(s,s')$ is \emph{mixed} if they are distinct.
Since the order of signals does not matter in the model, we state
our result for the mixed signals $s_{1}s_{2}$. Although the
posterior $p(\cdot|s_{1}s_{2})$ depends on the prior $p$ and the
experiment $\sigma$, we suppress this dependence in the notation. In contrast to Definition \ref{Def: Conf Updating},
we define confirmatory reading of mixed evidence as a statement about
a given agent.

\begin{defn}\label{Def: Conf Reading}(Confirmatory Reading
of Evidence) An agent with prior $p$ exhibits \emph{$\omega$-confirmatory
reading of mixed evidence} at $\sigma$ if 
\(
p(\omega|s_{1}s_{2})>p(\omega).
\)
There is \emph{no confirmatory reading} if $p(\omega|s_{1}s_{2})=p(\omega)$
for all $\omega\in\Omega$.\end{defn}

Being subjectively Bayesian, our agent exhibits $\omega_{1}$-confirmatory
reading if and only if $s_{1}s_{2}$ is more likely under the subjective experiment $\tau\sigma+(1-\tau)\rho^{p}$ in state
$\omega_{1}$ than $\omega_{2}$. This yields the
following proposition. For $\omega_{i}\in\Omega$, let $\omega_{-i}$
denote the state in $\Omega\setminus\{\omega_{i}\}$.

\begin{prop}\label{prop: confreading1} Consider any $\omega_{i}\in\Omega$.
An agent with prior $p$ satisfying our model \eqref{eqmain3} exhibits
$\omega_{i}$-confirmatory reading of mixed evidence at $\sigma$
if and only if 
\[
\left|\tau\sigma(s_{1}|\omega_{i})+(1-\tau)\rho^p(s_{1}|\omega_{i})-\frac{1}{2}\right|<\left|\tau\sigma(s_{1}|\omega_{-i})+(1-\tau)\rho^p(s_{1}|\omega_{-i})-\frac{1}{2}\right|.
\]

\end{prop}

Thus the reading is $\omega_{i}$-confirmatory if the subjective experiment
$\tau\sigma+(1-\tau)\rho^p$ yields $s_{1}$ conditional on $\omega_{i}$
with a likelihood closer to $\frac{1}{2}$ than that
conditional on $\omega_{-i}$. Two takeaways follow. First, while confirmatory updating arises even when all agents
share the same $\rho$, the only way to produce confirmatory reading
that reinforces opposing prior beliefs is for $\rho^p$ to depend on $p$. Indeed,
the displayed condition does not involve the prior. Second, confirmatory
bias and confirmatory reading are not equivalent in a model of limited
trust. Confirmatory updating can exist without confirmatory reading
-- for instance, when $\rho$ is independent of $p$ and $\rho(s_{1}\mid\omega_{1})=\rho(s_{1}\mid\omega_{2})=\frac{1}{2}$,
no agent exhibits confirmatory reading of mixed evidence under
symmetric experiments, but agents with strong priors exhibit
confirmatory updating. Similarly, confirmatory reading can exist without
confirmatory updating -- for instance, when $\frac{1}{2}-\rho(s_{1}|\omega_{2})>\rho(s_{1}|\omega_{1})-\frac{1}{2}>0$
then the agent exhibits $\omega_{1}$-confirmatory reading at all
symmetric $\sigma\in\Sigma^{1}$ but no confirmatory
updating at any nearly uninformative $\sigma\in\Sigma^{1}$.

\newpage

\appendix

\section{Projective Geometry: A Review}

\label{app:b}

\label{sec:Projective Geometry Review}

\subsection{Projective Spaces}

A \emph{line} in Euclidean space $\mathbb{R}^{m}$ is a one-dimensional
linear subspace $[x]=\{\lambda x:\lambda\in\mathbb{R}\}$ generated
by some non-zero vector $x\in\mathbb{R}^{m}\setminus\{0\}$. Under
this convention, a line passes through the origin. An \emph{affine
line} is a translate of a line and need not pass through the origin.
We use $[x]$ to denote a generic line and $v\in[x]$ to denote a
vector belonging to that line.

The \emph{projective space} $\bm{P}(\mathbb{R}^{m})$ \emph{associated
with} $\mathbb{R}^{m}$ is defined as the set of lines in $\mathbb{R}^{m}$,
\[
\bm{P}(\mathbb{R}^{m})=\{[x]:x\in\mathbb{R}^{m}\backslash\{0\}\}.
\]
The mapping $x\mapsto[x]$ is the ``canonical projection''. With
a slight abuse of notation, a line $[x]$ in $\mathbb{R}^{m}$ corresponds
to a \emph{projective point} $[x]$ in $\bm{P}(\mathbb{R}^{m})$.
Although projective spaces are not vector spaces, they have notions
of dimension and basis.\footnote{The dimension of $\bm{P}(\mathbb{R}^{m})$ is defined to be $m-1$.
For instance, if $m=1$ then $\bm{P}(\mathbb{R}^{m})$ contains one
projective point (corresponding to the line $\mathbb{R}$), and is
therefore zero-dimensional.}

A subset $L\subseteq\bm{P}(\mathbb{R}^{m})$ is a \emph{projective
line} if $L$ corresponds to the set of all one-dimensional subspaces
contained in a two-dimensional linear subspace in $\mathbb{R}^{m}$.
Similarly $L$ is a \emph{projective plane} if it corresponds to the
set of one-dimensional subspaces in a three-dimensional linear subspace
of $\mathbb{R}^{m}$. The projective line \emph{passing through $[x],[y]$}
is denoted by 
\[
[x]\vee[y]:=\{[z]\mid z\in\operatorname{Span}(\{x,y\})\setminus\{0\}\},
\]
which corresponds to the lines in the plane in $\mathbb{R}^{m}$ generated by $[x]$ and $[y]$. When $[x]\ne[y]$, this join is the projective
line through $[x]$ and $[y]$. When $[x]=[y]$, it is the
singleton $\{[x]\}$.

A subset $S\subseteq\bm{P}(\mathbb{R}^{m})$ is a \emph{subspace}
of $\bm{P}(\mathbb{R}^{m})$ if $[x],[y]\in S\implies[x]\vee[y]\subset S$,
that is, it contains the projective line passing through each pair
of projective points in $S$. For $X\subset\bm{P}(\mathbb{R}^{m})$,
a subset $S\subset X$ is a \textit{subspace} of $X$ if $S=\bar{S}\cap X$
where $\bar{S}$ is a subspace of $\bm{P}(\mathbb{R}^{m})$. We say
that $X\subset\bm{P}(\mathbb{R}^{m})$ is a \textit{locally projective
subspace} if (a) every subset $S\subseteq X$ satisfying 
\[
[x],[y]\in S\implies\bigl([x]\vee[y]\bigr)\cap X\subseteq S
\]
is a subspace of $X$; (b) for any projective line $L\subseteq\bm{P}(\mathbb{R}^{m})$
either $L\cap X=\emptyset$ or $L\cap X$ has at
least two points. Condition (a) means that, starting from any subset of $X$,
we obtain the smallest subspace of $X$ containing it by repeatedly
adding all points of $X$ lying on projective lines through pairs
of distinct points already included. This property is called being
\emph{generated by lines}. Condition (b) states that no ambient projective
line intersects $X$ in exactly one point.

The operation $\vee$ defined above gives the projective
line through two projective points. In Euclidean space, we analogously define the
affine line through two points $x,y$ by 
\(
x\vee y:=\{z\mid z=\alpha x+(1-\alpha)y\mid\alpha\in\mathbb{R}\}.
\)

\subsection{Morphisms}

Let $X\subseteq\bm{P}(\mathbb{R}^{m})$. For a function $\varphi:X\rightarrow\bm{P}(\mathbb{R}^{m})$,
for any $S\subseteq\bm{P}(\mathbb{R}^{m})$, denote $\varphi(S):=\varphi(S\cap X)$,
hence $\varphi(S)=\emptyset$ exactly for $S\subset X^{c}$. Similarly,
for $X'\subset\mathbb{R}^{m}$ and $\phi:X'\rightarrow\mathbb{R}^{m}$,
for any $S'\subset\mathbb{R}^{m}$, denote $\phi(S')=\phi(S'\cap X')$.

\begin{defn}\label{defmorph} Let $X\subseteq\bm{P}(\mathbb{R}^{m})$.
A function $\varphi:X\rightarrow\bm{P}(\mathbb{R}^{m})$ is a morphism
if for any $[x],[y]\in X$, 
\begin{align}
\varphi([x]\vee[y])\subseteq\varphi([x])\vee\varphi([y]).\label{eqmorphism}
\end{align}
Similarly, for $X'\subset\mathbb{R}^{m}$, a function $\phi:X'\rightarrow\mathbb{R}^{m}$
is a morphism if for any $x,y\in X'$, 
\begin{align}
\phi(x\vee y)\subseteq\phi(x)\vee\phi(y).\label{eqmorphism2}
\end{align}
\end{defn}

To interpret Definition \ref{defmorph}, take two points $[x],[y]\in X$
with distinct images. Then condition \eqref{eqmorphism}
requires every point of $X$ on the projective line $[x]\vee[y]$
to map to the projective line joining their images. If their images
coincide, every point of $X\cap([x]\vee[y])$ must map to that common
image. Condition \eqref{eqmorphism2} has the analogous interpretation
for affine lines in Euclidean space.

The notion of ``partial morphisms'' we define next allows us to state  the Fundamental Theorem of
Projective Geometry (Theorem \ref{Thm: FTPG}). This can be skipped without loss of continuity, however, since our main result will use a ``local''
version of the Fundamental Theorem (Theorem \ref{Thm: Local FTPG})
that does not refer to partial morphisms. A function $\varphi$
with domain $dom(\varphi)\subseteq\bm{P}(\mathbb{R}^{m})$ and image
$im(\varphi)\subseteq\bm{P}(\mathbb{R}^{m})$ is a
\emph{partial map}, denoted $\varphi:\bm{P}(\mathbb{R}^{m})\dasharrow\bm{P}(\mathbb{R}^{m})$.
For $[x]\notin\operatorname{dom}(\varphi)$, we write $\varphi([x])=\emptyset$.
We use the conventions 
\(
\emptyset\vee\{q\}=\{q\}\vee\emptyset=\{q\}\), and 
\(\emptyset\vee\emptyset=\emptyset,
\)
where $q$ is any projective point.

\begin{defn} A partial map $\varphi:\bm{P}(\mathbb{R}^{m})\dasharrow\bm{P}(\mathbb{R}^{m})$
is a partial morphism if for any $[x],[y]\in\bm{P}(\mathbb{R}^{m})$,
(\ref{eqmorphism}) holds true. \end{defn}

\subsection{Fundamental Theorem of Projective Geometry}

The ``Fundamental Theorem of Projective
Geometry'' is stated in various ways and degrees of generality,
and thus refers to a family of related results. The classic statements
involve \emph{collineations} (morphisms that are bijections whose
inverse is also a morphism), which are not suitable for our purposes.\footnote{We study updating rules that are morphisms, but not bijections; for
instance, Bayes' rule maps multiple likelihood vectors to the same
posterior, and thus is not a bijection.} The key generalization to (partial) morphisms is:

\begin{thm}[The Fundamental Theorem of Projective Geometry]\label{Thm: FTPG}
Suppose $m\ge3$. If $\varphi:\bm{P}(\mathbb{R}^{m})\dasharrow\bm{P}(\mathbb{R}^{m})$
is a partial morphism for which the image is not contained in a projective
line, then there exists an $m\times m$ matrix $M$ such that 
\(
\varphi([x])=[M\cdot x]\text{ for all }[x]\in dom(\varphi).
\)
Moreover, $M$ is unique up to multiplication by a non-zero scalar.\end{thm}

See, e.g., \citet[Theorem~3.1]{Faure2002} or \citet[Theorem~1.18]{SanchoDeSalas2026}. Those papers show that $\varphi$
is generated by a semilinear map on some vector space. But semilinear
maps between real vector spaces are linear, which yields the matrix
$M$.

This formulation is not directly applicable in our setup since we need a version of the Fundamental Theorem that allows for a bounded domain.\footnote{When $\varphi$ is a partial morphism, its domain must have particular
structure: the complement of the domain $\bm{P}(\mathbb{R}^{m})\backslash dom(\varphi)$ (or the so-called ``exceptional set'') must be a subspace
of $\bm{P}(\mathbb{R}^{m})$. The projective domain arising from the
likelihood vectors $(0,1)^{m}$ does not satisfy this requirement:
its complement contains two projective points whose join intersects the domain. Thus the exceptional set is not a subspace,
and the preceding theorem does not directly apply.} The following formulation, due to \citet[Theorems~1.18 and~3.11]{SanchoDeSalas2026}, weakens the domain requirement
on $\varphi$, but strengthens the lower bound on $m$ and the image
requirement.

\begin{thm}[The Fundamental Theorem of Locally Projective Geometry]\label{Thm: Local FTPG}
Suppose $m\ge4$. Let $P\subseteq\bm{P}(\mathbb{R}^{m})$ be a locally
projective subspace and let $\varphi:P\rightarrow\bm{P}(\mathbb{R}^{m})$
be a morphism whose image is not contained in a plane. Then there
exists an $m\times m$ matrix $M$ such that
\(
\varphi([x])=[M\cdot x]\text{ for all }[x]\in P.
\)
Moreover, $M$ is unique up to multiplication by a nonzero scalar.\end{thm} 
\begin{proof}
By condition (a) in the definition of a locally projective subspace
and Remark 1.3 of \citet{SanchoDeSalas2026}, $\varphi$ is a morphism
in his sense. By condition (b), every projective line through a point
$[x]\in P$ contains another point of $P$. By his Proposition 3.7
and its proof, $P$ is locally projective, and the inclusion $P\hookrightarrow\bm{P}(\mathbb{R}^{m})$
is a minimal projective embedding in the sense of his Theorem 3.9.
This projective embedding is already of the form required in his Remark
3.10, with $V=\mathbb{R}^{m}$. Since $m\ge4$, his Theorem 3.11 therefore
yields a unique extension $\hat{\varphi}:\bm{P}(\mathbb{R}^{m})\dasharrow\bm{P}(\mathbb{R}^{m})$.
Applying Theorem \ref{Thm: FTPG} above gives the required matrix,
whose uniqueness up to a nonzero scalar follows from the uniqueness
of the extension and the uniqueness assertion in that theorem. 
\end{proof}

Early precursors of the local result are \citet{Lenz1958} and \citet{Lowen1982}.\footnote{\citet[Hilfssatz~3, p.~348]{Lenz1958} states and proves that ``Jede geradentreue
Abbildung eines abgeschlossenen Simplex laesst sich, wenn nicht alle
Bildpunkte in einer Geraden liegen, eindeutig zu einer -- unter Umstaenden
ausgearteten -- projektiven Kollineation des Raumes fortsetzen,''
that is, ``Any line-preserving map on a closed simplex, provided
that its image is not contained in a line, extends uniquely to a projective
collineation of the ambient space, possibly a degenerate one.'' When
stated for Euclidean spaces, \citet{Lowen1982} proves that for $m\ge3$
and any open subset $U\subset\mathbb{R}^{m}$, a continuous injective
morphism $\varphi:\{[u]:u\in U\}\rightarrow\mathbf{P}(\mathbb{R}^{m})$
extends uniquely to a morphism on $\bm{P}(\mathbb{R}^{m})$.}
\subsection{A Representation for Local Morphisms}
 
\label{sec:appc}
 
The results above are stated for projective spaces. We now use them
to prove a representation result for functions on Euclidean space, which is applicable to our axiomatic characterization. To link the points of $\mathbb{R}^{n}$ to the points of a projective
space, define $\pi:\mathbb{R}^{n}\rightarrow\mathbb{R}^{n+1}\backslash\{0\}$
by mapping each $v=(v_{1},...,v_{n})\in\mathbb{R}^{n}$ to $\pi(v)=(v_{1},...,v_{n},1)\in\mathbb{R}^{n+1}\backslash\{0\}$.
We call this the ``homogeneous-coordinate representation'' for
$v\in\mathbb{R}^{n}$ and, as conventionally done, write it using colons:
\(
\pi(v)=(v_{1}:...:v_{n}:1).
\)
This embeds $\mathbb{R}^{n}$ in the
hyperplane $H\subset\mathbb{R}^{n+1}$ defined by vectors of the form
$(v_{1},...,v_{n},1)$. Each $v\in\mathbb{R}^{n}$ then corresponds
to the projective point $[\pi(v)]\in\bm{P}(\mathbb{R}^{n+1})$, that
is, to the line in $\mathbb{R}^{n+1}$ passing through the origin
and $\pi(v)$. Conversely, each $\tilde{v}=(v_{1}:...:v_{n}:v_{n+1})$
with $v_{n+1}\ne0$ lies on exactly one such line, and the point of
$\mathbb{R}^{n}$ that it corresponds to is its so-called ``Cartesian
coordinate'' representation
\(
\left(\frac{v_{1}}{v_{n+1}},\ldots,\frac{v_{n}}{v_{n+1}}\right).
\)\footnote{The homogeneous coordinates of the form $[v_{1}:\cdots:v_{n}:0]$
do not correspond to points in Cartesian space. Rather, they correspond
to \emph{directions} in the Euclidean space $\mathbb{R}^{n}$ and
are interpreted geometrically as points ``at infinity.'' Because
such projective points are included in projective space, any two distinct projective lines in a common projective plane intersect (parallel lines meet at a projective point at infinity).}
 
Using a proof strategy familiar from applications of the Fundamental Theorem, we project a function $f$ on Euclidean space
to a function $\varphi$ on a locally projective subset of $\bm{P}(\mathbb{R}^{n+1})$,
where $\varphi$ maps $[\pi(v)]$ to $[\pi(f(v))]$. Note that $f$
is defined on a subset of $\mathbb{R}^{n}$ and, instead of $\bm{P}(\mathbb{R}^{n})$,
we define $\varphi$ on $\bm{P}(\mathbb{R}^{n+1})$. Applying the Fundamental
Theorem of Locally Projective Geometry (Theorem \ref{Thm: Local FTPG})
yields a representation for $\varphi$, which induces a representation
for $f$. The result applies to functions $f:C\rightarrow\mathbb{R}^{n}$
in the following class:
 
\begin{defn}(Local Morphism) A mapping $f:C\rightarrow\mathbb{R}^{n}$
is a local morphism if (a) $\emptyset\ne C\subset\mathbb{R}^{n}$ is open; (b) $f:C\rightarrow\mathbb{R}^{n}$ is a morphism; and (c) $f(C)$ is not contained in a plane in $\mathbb{R}^{n}$.
\end{defn}

\begin{thm} \label{thmrepresentation} Suppose $n\ge4$. If $f:C\rightarrow\mathbb{R}^{n}$
is a local morphism then there exist $c\in\mathbb{R},\;a,b\in\mathbb{R}^{n},\;M\in\mathbb{R}^{n\times n}$
such that
\[
f(v)=\frac{M\cdot v+a}{b^{\top}\cdot v+c},\qquad v\in C.
\]
 
\end{thm}
\begin{proof}
Let
\(
P=\{[\pi(v)]\in\bm{P}(\mathbb{R}^{n+1}):v\in C\},
\)
which corresponds to all lines in $\mathbb{R}^{n+1}$ through
the embedding of $C$ in $H$. Endow $P$ with the family of subspaces
of the form $S\cap P$ where $S$ is a subspace of $\bm{P}(\mathbb{R}^{n+1})$.
Since $v\mapsto[\pi(v)]$ is one-to-one,
\(
\varphi([\pi(v)])=[\pi(f(v))],
\)
$v\in C$, defines a map $\varphi:P\rightarrow\bm{P}(\mathbb{R}^{n+1})$.
We verify the hypotheses of Theorem \ref{Thm: Local FTPG}.
 
(i) \emph{$P$ is a locally projective subspace of $\bm{P}(\mathbb{R}^{n+1})$.}
Since $C$ is open in $\mathbb{R}^{n}$, $P$ is open in $\bm{P}(\mathbb{R}^{n+1})$.
By Remark 3.13(b) of \citet{SanchoDeSalas2026}, $P$ is generated by
lines, which verifies condition (a) of a locally
projective subspace. For any projective line $L$, $L\cap P$ is open in $L$ and, if nonempty, contains
infinitely many points, verifying condition (b).
 
(ii) \emph{$\varphi(P)$ is not contained in a projective plane.}
Since $f(C)$ is not contained in an affine plane, it contains four
affinely independent points $y_{0},y_{1},y_{2},y_{3}$. Hence $\pi(y_{0}),...,\pi(y_{3})$ are linearly independent
in $\mathbb{R}^{n+1}$. A projective plane corresponds to a three-dimensional
linear subspace of $\mathbb{R}^{n+1}$ and therefore cannot contain
all four points $[\pi(y_{i})]\in\varphi(P)$.
 
(iii) \emph{$\varphi$ is a morphism.} Take any $u,v,w\in C$ such that
\(
[\pi(w)]\in[\pi(u)]\vee[\pi(v)].
\)
There exist $\alpha,\beta\in\mathbb{R}$ such that
\(
\pi(w)=\alpha\pi(u)+\beta\pi(v).
\)
Comparing the last coordinates gives $\alpha+\beta=1$, so $w=\alpha u+(1-\alpha)v$.
Since $f$ is a morphism, $f(w)=\theta f(u)+(1-\theta)f(v)$ for some
$\theta\in\mathbb{R}$, so $\pi(f(w))=\theta\pi(f(u))+(1-\theta)\pi(f(v))$.
Hence
\(
\varphi([\pi(w)])\in\varphi([\pi(u)])\vee\varphi([\pi(v)]).
\)
 
By (i)--(iii) and $n\ge4$, Theorem \ref{Thm: Local FTPG} with $m=n+1$
yields an $(n+1)\times(n+1)$ matrix $\tilde{M}$ such that $\varphi([\tilde{v}])=[\tilde{M}\cdot\tilde{v}]$
for all $[\tilde{v}]\in P$. Write $\tilde{M}=\begin{pmatrix}M & a\\[3pt]
b^{\top} & c
\end{pmatrix}$, where $c\in\mathbb{R},\;a,b\in\mathbb{R}^{n},\;M\in\mathbb{R}^{n\times n}$.
Take any $v\in C$. Since $[\pi(f(v))]=[\tilde{M}\cdot\pi(v)]$, there
exists a scalar $\lambda\ne0$ such that
\(
(f(v),1)=\pi(f(v))=\lambda\tilde{M}\cdot\pi(v)=\lambda(M\cdot v+a,\;b^{\top}\cdot v+c).
\)
Comparing the last coordinates gives $\lambda=\frac{1}{b^{\top}\cdot v+c}$,
and comparing the first $n$ coordinates gives $f(v)=\frac{M\cdot v+a}{b^{\top}\cdot v+c}$,
as was to be shown.
\end{proof}

\begin{thm}\label{thmrepresentation2}
Let $C\subset\mathbb{R}^n$ be open and convex, and let
$f:C\rightarrow\mathbb{R}^n$ have an image not contained in a plane.
If $f$ preserves collinearity, in the sense that for all $x,y\in C$
and $\alpha\in[0,1]$, there exists $\gamma\in[0,1]$ such that
$f(\alpha x+(1-\alpha)y)=\gamma f(x)+(1-\gamma)f(y)$,
then $f$ is a morphism. Moreover, if $n\geq4$, then $f$ admits
a representation as in Theorem~\ref{thmrepresentation}.
\end{thm}

\begin{proof}
For any collinear $u,v,t\in C$, we first note that
\begin{align}
f(u)\ne f(v)
\quad\Longrightarrow\quad
f(t)\in f(u)\vee f(v).
\label{eq:collinearity-distinct-images}
\end{align}
This is immediate when $t$ lies between $u$ and $v$.
If, say, $u$ lies between $v$ and $t$, then
$f(u)=\gamma f(t)+(1-\gamma)f(v)$ with $\gamma>0$ because
$f(u)\ne f(v)$; solving for $f(t)$ gives the same conclusion.

We must show that $f(C\cap(x\vee y))\subseteq f(x)\vee f(y)$
for all $x,y\in C$. The case $x=y$ is trivial, and
\eqref{eq:collinearity-distinct-images} handles $f(x)\ne f(y)$.
Suppose, towards a contradiction, that $x\ne y$, $f(x)=f(y)$,
and some $z\in C\cap(x\vee y)$ satisfies $f(z)\ne f(x)$.
Since $f(C)$ is not contained in a line, choose $w\in C$ with
$f(w)\notin f(x)\vee f(z)$. By
\eqref{eq:collinearity-distinct-images},
$w\notin x\vee z=x\vee y$.
Let $P$ be the plane through $x,y,w$, and let
$S:=\{v\in C\cap P:f(v)\in f(x)\vee f(w)\}$.
All interiors below are relative to $P$. We will show that $x,w$
are interior points of $S$ and that $f$ is constant on this
interior, contradicting $f(x)\ne f(w)$.

Since $f$ preserves collinearity, $S$ is convex, and
\eqref{eq:collinearity-distinct-images} implies that it contains
$C\cap(x\vee w)$ and $C\cap(y\vee w)$.
These are relatively open intervals on two distinct lines
through $w$, both containing $w$. Their convex hull is therefore
relatively open in $P$ and is contained in $S$.
Since both $x$ and $w$ belong to this hull, they are interior
points of $S$.

Next, any two points of $S$ on a line through a point of
$(C\cap P)\setminus S$ must have the same image. Otherwise,
\eqref{eq:collinearity-distinct-images} places the image of
that outside point on $f(x)\vee f(w)$, contradicting its
exclusion from $S$.

The complement $(C\cap P)\setminus S$ cannot be contained
in a line $L\subset P$. Otherwise $z\in L$, and openness
of $C\cap P$ in $P$ gives a short segment centred at $z$
with both endpoints in $C\cap P$ but off $L$.
These endpoints belong to $S$, so convexity would force
$z\in S$, a contradiction.

Fix an interior point $u$ of $S$. Since the complement is
not contained in $u\vee z$, choose $z'$ in this complement
but outside $u\vee z$. To compare $f(u)$ with $f(v)$ for $v$
near $u$, let $r$ be the intersection of $u\vee z$ and
$v\vee z'$. At $v=u$, these lines are nonparallel and meet
at $u$. Their intersection therefore remains well-defined
and tends to $u$ as $v\to u$. Because $u$ is interior to $S$,
both $r$ and $v$ belong to $S$ when $v$ is sufficiently close
to $u$ in $P$.

If $f(u)\ne f(r)$, \eqref{eq:collinearity-distinct-images}
gives $f(z)\in f(u)\vee f(r)=f(x)\vee f(w)$,
contradicting $z\notin S$. Hence $f(u)=f(r)$.
Likewise, if $f(r)\ne f(v)$,
\eqref{eq:collinearity-distinct-images} gives
$f(z')\in f(r)\vee f(v)=f(x)\vee f(w)$,
contradicting $z'\notin S$. Thus $f(u)=f(r)=f(v)$.

Thus $f$ is locally constant on the interior of $S$, which is
convex and hence connected, so $f$ is constant there.
In particular, $f(x)=f(w)$, contradicting our choice of $w$.
Therefore $f$ is a morphism. When $n\geq4$, the representation
follows from Theorem~\ref{thmrepresentation}.
\end{proof}
\section{Proof of Theorems \ref{thmmain1}-\ref{thmmain4}}

Let $n:=|\Omega|\geq4$. Recall that an experiment $\sigma$ is a
mapping $\Omega\to\Delta S$. For each $s\in S$, let 
\(
\sigma_{s}:=\bigl(\sigma(s\mid\omega)\bigr)_{\omega\in\Omega}\in\mathbb{R}^{\Omega}
\)
denote the likelihood vector of signal $s$ under experiment $\sigma$.
We also write $\sigma_{s\omega}$ as shorthand for $\sigma(s\mid\omega)$.

The following preference-based analog of Theorem~\ref{thmmain4belief}
is the main step toward Theorems~\ref{thmmain1} and~\ref{thmmain2}.
It requires only minimal additional assumptions beyond Axiom~\ref{a1}:
Axiom~\ref{a5} is a technical richness condition, and Axiom~\ref{a7}
below is a weakening of Axiom~\ref{a2} which states that $\succsim_{\sigma,s}$
depends on $\sigma$ only through the likelihood vector $\sigma_{s}$.

\begin{ax}[Likelihood Consistency, LC]\label{a7} If $\sigma_{s}=\sigma'_{s}$,
then $\succsim_{\sigma,s}=\succsim_{\sigma',s}$. \end{ax}

\begin{thm}\label{thmmain4} $\{\succsim_{0},\succsim_{\sigma,s}\}_{(\sigma,s)\in(\Sigma\times S)}$
satisfy Axioms~\ref{a1}, \ref{a5}, and \ref{a7} if and only if
there exist a quasi-experiment $\rho$, $\tau\in(0,1]$, and a family
of $|\Omega|\times|\Omega|$ matrices $(M_{s})_{s\in S}$ such that
\begin{align}
p_{\sigma}(\omega\mid s)=\frac{[\tau M_{s}\cdot\sigma_{s}+(1-\tau)\rho_{s}]_{\omega}\cdot p(\omega)}{\sum_{\omega'}[\tau M_{s}\cdot\sigma_{s}+(1-\tau)\rho_{s}]_{\omega'}\cdot p(\omega')},
\end{align}
where $\operatorname{rank}[\tau M_{s}\ (1-\tau)\rho_{s}]=|\Omega|$
for each $s\in S$. \end{thm}

\subsection{Proof of Theorem \ref{thmmain4}}

Fix the family of binary relations over $\mathcal{F}$: 
\(
\{\succsim_{0},\succsim_{\sigma,s}\}_{(\sigma,s)\in\Sigma\times S}
\)
that satisfy Axiom \ref{axiom:seu}. By Axiom \ref{axiom:seu}, for
some non-constant and affine $u:X\rightarrow\mathbb{R}$, $\succsim_{0}$
admits a representation $U(f)=\int_{\Omega}u\circ f\,dp,$ and for
each $(\sigma,s)\in\Sigma\times S$, $\succsim_{\sigma,s}$ admits
a representation $U_{\sigma,s}(f)=\int_{\Omega}u\circ f\,dp_{\sigma}(\cdot\mid s)$.
In later proofs, $\succsim^{E}_{\sigma,s}$ denotes
the conditional preference on an event $E\subseteq\Omega$, defined
as in \citet{Ghirardato2002}: for $f,g\in\mathcal{F}$, write $f\succsim^{E}_{\sigma,s}g$
if there exists $h\in\mathcal{F}$ such that 
\(
fEh\succsim_{\sigma,s}gEh,
\)
where $fEh$ agrees with $f$ on $E$ and with $h$ on
$\Omega\setminus E$, and $gEh$ is defined analogously.

By Axiom \ref{a7} (Likelihood consistency), if $\sigma_{s}=\sigma'_{s}$,
then $p_{\sigma}(\cdot\mid s)=p_{\sigma'}(\cdot\mid s)$. Therefore,
for each $\sigma_{s}\in\Sigma_{s}$ and $s\in S$, there exists a
distribution $p_{\sigma_{s}}(\cdot\mid s)\in\Delta(\Omega)$ such
that $\succsim_{\sigma,s}$ is represented by 
\(
U_{\sigma,s}(f)=\int_{\Omega}u\circ f\,dp_{\sigma_{s}}(\cdot\mid s).
\)
Without loss of generality, assume $0$ lies in the interior of $u(X)$.
Fix $s\in S$. Define $\phi_{s}:\Sigma_{s}\rightarrow\Delta(\Omega)\subset\mathbb{R}^{n}$
by 
\(
\phi_{s}(\sigma_{s}):=p_{\sigma_{s}}(\cdot\mid s).
\)

We first claim that the image of $\phi_{s}$ is not contained in a
plane in $\mathbb{R}^{n}$. Suppose, towards a contradiction, that it is. Since $n:=|\Omega|\geq4$, the affine hull of the probability
simplex $\Delta(\Omega):=\{v\in\mathbb{R}^{n}_{+}\mid\sum^{n}_{m=1}v_{m}=1\}$
is at least 3-dimensional. Since a plane has lower dimension than the simplex, there exist $f^{*}\in\mathcal{F}$ and $c\in\mathbb{R}$
such that $u_{f^{*}}\notin\operatorname{Span}\{1_{\Omega}\}$ and
\(
p_{\sigma_{s}}(\cdot\mid s)\cdot u_{f^{*}}=c
\)
for all $\sigma\in\Sigma$. Hence, for all $\sigma,\sigma'\in\Sigma$ and all $x\in X$, 
\(
f^{*}\succ_{\sigma,s}x\iff f^{*}\succ_{\sigma',s}x.
\)
Since $u_{f^{*}}\notin\operatorname{Span}\{1_{\Omega}\}$, there exist
$\omega,\psi\in\Omega$ such that $f^{*}(\omega)\not\sim_{0}f^{*}(\psi)$,
contradicting Axiom \ref{a5} (Richness).

\begin{lem}\label{lemmaD1-1} $\phi_{s}$ preserves collinearity.
\end{lem} 
\begin{proof}
Let $\mu:=\phi_{s}(\sigma_{s})$, $\mu':=\phi_{s}(\sigma'_{s})$,
and $\mu'':=\phi_{s}(\sigma''_{s})$, where $\sigma''=\alpha\sigma+(1-\alpha)\sigma'$
for some $\alpha\in(0,1)$. It suffices to show that $\mu''\in\operatorname{co}\{\mu,\mu'\}$. For each $f\in\mathcal{F}$, let $u_{f}\in\mathbb{R}^{n}$ denote
the utility act. Since $\succsim_{\sigma,s}$ is represented by $\mu$,
we have $f\succsim_{\sigma,s}g\iff\mu\cdot(u_{f}-u_{g})\ge0,$ and
similarly for $\sigma'$
and $\sigma''$. Thus Axiom \ref{a1} implies that for all $f,g\in\mathcal{F}$, $\mu\cdot(u_{f}-u_{g})\ge0$
and $\mu'\cdot(u_{f}-u_{g})\ge0$ together imply $\mu''\cdot(u_{f}-u_{g})\ge0.$
Because $u$ is non-constant and affine and $X$ is convex, $u(X)$
contains a non-degenerate interval. Hence, for every $v\in\mathbb{R}^{n}$,
there exist acts $f,g\in\mathcal{F}$ and $\lambda>0$ such that 
\(
u_{f}-u_{g}=\lambda v.
\)
Therefore, for every $v\in\mathbb{R}^{n}$,
if $\mu\cdot v\ge0$ and $\mu'\cdot v\ge0$, then $\mu''\cdot v\ge0$.

Now suppose $\mu''\notin\operatorname{co}\{\mu,\mu'\}$. Since $\operatorname{co}\{\mu,\mu'\}$
is closed and convex, the separating hyperplane theorem yields $v\in\mathbb{R}^{n}$
and $c\in\mathbb{R}$ such that 
\(
\mu\cdot v\ge c, \mu'\cdot v\ge c, \text{ and }\mu''\cdot v<c.
\) Let $\tilde{v}:=v-c1_{\Omega}$. Since $\mu,\mu',\mu''\in\Delta\Omega$,
we have 
\(
\mu\cdot\tilde{v}\ge0, \mu'\cdot\tilde{v}\ge0, \text{ and }\mu''\cdot\tilde{v}<0,
\)
contradicting the preceding implication. Hence $\mu''\in\operatorname{co}\{\mu,\mu'\}$ and $\phi_{s}$ preserves collinearity.
\end{proof}

By Theorem~\ref{thmrepresentation2}, since $\phi_s$ preserves collinearity and its image is not contained in a plane, there
exist $A_s\in\mathbb{R}^{n\times n}$,
$b_s\in\mathbb{R}^n$, $c_s\in\mathbb{R}^{1\times n}$,
and $d_s\in\mathbb{R}$ such that
\[
\phi_s(\zeta)=\frac{A_s\zeta+b_s}{c_s\zeta+d_s},
\qquad \zeta\in\Sigma_s.
\]
Since $c_s\zeta+d_s$ is nonzero on the convex set
$\Sigma_s$, it has a constant sign. Multiplying all
coefficients by $-1$ if necessary, we may assume
$c_s\zeta+d_s>0$ throughout $\Sigma_s$.
Since $\phi_s(\zeta)\in\Delta(\Omega)$, we have
\(
A_s\zeta+b_s\geq0,
\)
and
\(
c_s\zeta+d_s
=\mathbf{1}^{\top}(A_s\zeta+b_s)
=\|A_s\zeta+b_s\|_1.
\)
Therefore,
\[
\phi_s(\zeta)
=\frac{A_s\zeta+b_s}{\|A_s\zeta+b_s\|_1}.
\]
Finally, since $0\in\overline{\Sigma_s}$ and
$A_s\zeta+b_s\geq0$ for all $\zeta\in\Sigma_s$,
taking limits yields $b_s\geq0$.

We next rewrite the representation in the form stated in Theorem \ref{thmmain4}. For $v,w\in\mathbb{R}^n$, let $v\odot w\equiv(v_iw_i)_{i=1}^n$ denote their componentwise product. If $b_{s}=0$ for
all $s\in S$, then take $M_{s}:=D^{-1}_{p}A_{s}$ for all $s\in S$,\footnote{Since $p$ lies in the relative interior of $\Delta(\Omega)$,
each component of $p$ is strictly positive. Thus the diagonal matrix
$D_{p}:=\operatorname{diag}(p)$ is invertible.}
and 
\(
\phi_{s}(\sigma_{s})=\frac{p\odot[M_{s}\sigma_{s}]}{{\|p\odot[M_{s}\sigma_{s}]\|_{1}}}
\)
for all $\sigma\in\Sigma$ and $s\in S$, yielding the representation
in Theorem \ref{thmmain4} with weight zero on the quasi-experiment
$\rho$. Now suppose $b_{s}\neq0$ for some $s\in S$.

\paragraph{Step 1. Constructing $\tau$, $\rho$, and $M_{s}$}

Choose $\alpha>0$ such that $\alpha\leq\min_{\omega:\sum_{s}b_{s\omega}\neq0}\frac{p_{\omega}}{\sum_{s}b_{s\omega}}$
and define 
\(
\rho_{s}:=\alpha D^{-1}_{p}b_{s}
\)
and 
\(M_{s}:=D^{-1}_{p}A_{s}
\)
for all $s\in S$. Let 
\(
\tau:=\frac{\alpha}{1+\alpha}\in(0,1).
\)
Then $1-\tau=\frac{1}{1+\alpha}$ and $(1-\tau)\rho_{s}=\frac{\alpha}{1+\alpha}D^{-1}_{p}b_{s}=\tau D^{-1}_{p}b_{s}$.
Therefore, 
\[
\tau M_{s}\sigma_{s}+(1-\tau)\rho_{s}=\tau D^{-1}_{p}A_{s}\sigma_{s}+\tau D^{-1}_{p}b_{s}=\tau D^{-1}_{p}(A_{s}\sigma_{s}+b_{s}).
\]
Multiplying componentwise by $p$, we obtain 
\(
p\odot[\tau M_{s}\sigma_{s}+(1-\tau)\rho_{s}]=\tau(A_{s}\sigma_{s}+b_{s}).
\)
Since normalization is invariant under multiplication by a positive
scalar, 
\[
\frac{p\odot[\tau M_{s}\sigma_{s}+(1-\tau)\rho_{s}]}{\|p\odot[\tau M_{s}\sigma_{s}+(1-\tau)\rho_{s}]\|_{1}}=\frac{A_{s}\sigma_{s}+b_{s}}{\|A_{s}\sigma_{s}+b_{s}\|_{1}}.
\]
Since $A_{s}\sigma_{s}+b_{s}\geq0$, the $l_{1}$ norm equals the
coordinate sum, as in Theorem \ref{thmmain4}.

\paragraph{Step 2. Verifying that $\rho$ is a quasi-experiment}

Recall that $\rho_{s}=\alpha D^{-1}_{p}b_{s}$ for all $s\in S$.
Since $b_{s}\ge0$ for every $s\in S$ and $p_{\omega}>0$ for every
$\omega\in\Omega$, we have 
\(
\rho_{s\omega}=\alpha\frac{b_{s\omega}}{p_{\omega}}\ge0
\)
for all $s\in S,\ \omega\in\Omega.$ 

It remains to verify the row-sum condition. For each $\omega\in\Omega$,
we have $\sum_{s\in S}\rho_{s\omega}=\alpha\sum_{s\in S}\frac{b_{s\omega}}{p_{\omega}}=\alpha\frac{\sum_{s\in S}b_{s\omega}}{p_{\omega}}.$ Hence, if $\sum_{s}b_{s\omega}>0$ then $\sum_{s\in S}\rho_{s\omega}\le1$ by our choice of $\alpha$,
and if $\sum_{s}b_{s\omega}=0$ then $b_{s\omega}=0$ for every $s$
so that $\sum_{s\in S}\rho_{s\omega}=0\le1.$ Therefore $\rho_{s\omega}\geq 0$ and $\sum_{s\in S}\rho_{s\omega}\le1$ for all $s\in S$, $\omega\in \Omega$. Finally, since $b_{s}\neq0$ for some $s\in S$, it follows that $\rho\neq0$.
This shows that $\rho$ is a quasi-experiment. 

\paragraph{Step 3. Verifying the Rank Condition}

It remains to show that
$\operatorname{rank}[\,\tau M_s\ \ (1-\tau)\rho_s\,]=n$ for all
$s\in S$. Suppose, towards a contradiction, that the rank is less than
$n$ for some $s\in S$. We construct an act with nonconstant payoffs
whose conditional expected utility at $s$ is the same under every
experiment, which Axiom~\ref{a5} rules out. Since $[\,\tau M_s\ \ (1-\tau)\rho_s\,]$ has $n$ rows, there exists a nonzero $w\in\mathbb{R}^{n}$ such
that $w^{\top}M_s=0$ and $(1-\tau)w^{\top}\rho_s=0$. Let
$v:=D_p^{-1}w$, which is nonzero. By the representation established
above,
\(
\phi_s(\zeta)=D_p[\tau M_s\zeta+(1-\tau)\rho_s]/d(\zeta),
\)
where $d(\zeta)>0$ is the normalizing constant. Hence, for all
$\zeta\in\Sigma_s$,
\(
v^{\top}\phi_s(\zeta)=w^{\top}[\tau M_s\zeta+(1-\tau)\rho_s]/d(\zeta)=0.
\)
Moreover, $v\notin\operatorname{span}\{1_{\Omega}\}$: if
$v=c1_{\Omega}$ for some $c\neq0$, then $v^{\top}\phi_s(\zeta)=c\neq0$
because $\phi_s(\zeta)\in\Delta(\Omega)$.

Since $u(X)$ contains a nondegenerate interval, there exist $a>0$,
$b\in\mathbb{R}$, and $f\in\mathcal{F}$ such that
$u_f=av+b1_{\Omega}$. For every $\sigma\in\Sigma$,
\(
U_{\sigma,s}(f)=u_f^{\top}\phi_s(\sigma_s)=a\,v^{\top}\phi_s(\sigma_s)+b=b.
\)
On the other hand, since $v\notin\operatorname{span}\{1_{\Omega}\}$,
$u_f$ is nonconstant, so $f_{\omega}\not\sim_{0}f_{\psi}$ for some
$\omega,\psi\in\Omega$. Axiom~\ref{a5} then yields
$\sigma,\sigma'\in\Sigma$ and $x\in X$ with $f\succ_{\sigma,s}x$ and
$f\prec_{\sigma',s}x$, that is,
\(
U_{\sigma,s}(f)>u(x)>U_{\sigma',s}(f),
\)
a contradiction.
\subsection{Proof of Theorem \ref{thmmain2}}
 
\begin{lem}\label{lemma2} For $\sigma,\sigma'\in\Sigma$, $\omega,\omega'\in\Omega$,
and $s\in S$, suppose $\sigma(s\mid\omega)=\sigma'(s\mid\omega)$
and $\sigma(s\mid\omega')=\sigma'(s\mid\omega')$. If Axiom \ref{a2} holds, $p_{\sigma_{s}}(\omega'\mid s)>0$, and $p_{\sigma'_{s}}(\omega'\mid s)>0$, then
\(
{p_{\sigma_{s}}(\omega\mid s)}/{p_{\sigma_{s}}(\omega'\mid s)}={p_{\sigma'_{s}}(\omega\mid s)}/{p_{\sigma'_{s}}(\omega'\mid s)}.
\)
\end{lem} 
\begin{proof}
Let $E=\{\omega,\omega'\}$. Since $\sigma(s\mid\psi)=\sigma'(s\mid\psi)$
for all $\psi\in E$, Axiom \ref{a2} implies that the conditional
preferences $\succsim^{E}_{\sigma,s}$ and $\succsim^{E}_{\sigma',s}$
coincide, which are represented by $f\mapsto\sum_{\psi\in E}p_{\sigma_{s}}(\psi\mid s)\,u(f(\psi))$
and $f\mapsto\sum_{\psi\in E}p_{\sigma'_{s}}(\psi\mid s)\,u(f(\psi)),$
respectively. Since both use the same utility index $u$, the induced posterior odds on $E$ coincide,
as desired. 
\end{proof}

\begin{lem}\label{lemma:Ms-diagonal} Suppose that Axiom \ref{a2} holds
and that, for each $s\in S$ and $\sigma\in\Sigma$, $p_{\sigma_{s}}(\cdot\mid s)$
is the normalization of $p\odot[\tau M_{s}\sigma_{s}+(1-\tau)\rho_{s}]$,
where $\tau\in(0,1]$ and $\operatorname{rank}[\,\tau M_{s}\ \ (1-\tau)\rho_{s}\,]=n$.
Then, for every $s\in S$, the matrix $M_{s}$ is diagonal. \end{lem}
 
\begin{proof}
Fix $s\in S$. Recall that $\sigma_{s}=(\sigma_{s\psi})_{\psi\in\Omega}\in\Sigma_{s}\subset\mathbb{R}^{n}$
is a vector indexed by states. Fix two distinct states $\omega,\omega'\in\Omega$.
We first show that, for every $\kappa\notin\{\omega,\omega'\}$, 
\(
(M_{s})_{\omega\kappa}=(M_{s})_{\omega'\kappa}=0.
\)
By Lemma \ref{lemma2}, the ratio 
\(
R_{\omega,\omega'}(\sigma_{s}):={p_{\sigma_{s}}(\omega\mid s)}/{p_{\sigma_{s}}(\omega'\mid s)}
\)
depends only on the two coordinates $\sigma_{s\omega}$ and $\sigma_{s\omega'}$ wherever its denominator is positive. Using the representation above, 
\[
R_{\omega,\omega'}(\sigma_{s})={p_{\omega}\left[\tau\sum_{\psi\in\Omega}(M_{s})_{\omega\psi}\sigma_{s\psi}+(1-\tau)\rho_{s\omega}\right]}\Bigg/{p_{\omega'}\left[\tau\sum_{\psi\in\Omega}(M_{s})_{\omega'\psi}\sigma_{s\psi}+(1-\tau)\rho_{s\omega'}\right]}.
\]
The denominator is positive on a nonempty open subset of $\Sigma_{s}$:
otherwise $p_{\sigma_{s}}(\omega'\mid s)=0$ on a nonempty open set,
so $\phi_{s}(\Sigma_{s})$ would be locally contained in the hyperplane
$\{x\in\Delta\Omega:x_{\omega'}=0\}$, contradicting the full-rank
condition established above.
 
Now fix $\kappa\notin\{\omega,\omega'\}$. Since $\Sigma_{s}$ has
nonempty interior, we may vary $\sigma_{s\kappa}$ on a nonempty open
set while holding the other coordinates fixed. On the intersection
with the open set where the denominator is positive, $R_{\omega,\omega'}$
is well-defined and, by the preceding paragraph, independent of $\sigma_{s\kappa}$.
Hence 
\(
{\partial R_{\omega,\omega'}(\sigma_{s})}/{\partial\sigma_{s\kappa}}=0
\)
on a full-dimensional open set.
 
Since $p_{\omega}$ and $p_{\omega'}$ are positive constants, and $\tau>0$, the
quotient rule yields 
\[
(M_{s})_{\omega\kappa}\left[\tau\sum_{\psi\in\Omega}(M_{s})_{\omega'\psi}\sigma_{s\psi}+(1-\tau)\rho_{s\omega'}\right]=(M_{s})_{\omega'\kappa}\left[\tau\sum_{\psi\in\Omega}(M_{s})_{\omega\psi}\sigma_{s\psi}+(1-\tau)\rho_{s\omega}\right].
\]
This identity holds on a full-dimensional open set. Since both sides
are affine functions of $\sigma_{s}\in\mathbb{R}^{\Omega}$, their
coefficients must coincide. Hence, for every $\psi\in\Omega$, 
\[
(M_{s})_{\omega\kappa}(M_{s})_{\omega'\psi}=(M_{s})_{\omega'\kappa}(M_{s})_{\omega\psi} \ \  \text{ and } \ \ (1-\tau)(M_{s})_{\omega\kappa}\rho_{s\omega'}=(1-\tau)(M_{s})_{\omega'\kappa}\rho_{s\omega}.
\]
Therefore, if either $(M_{s})_{\omega\kappa}\neq0$ or $(M_{s})_{\omega'\kappa}\neq0$,
then the $\omega$-row and the $\omega'$-row of 
\(
[\,\tau M_{s}\ \ (1-\tau)\rho_{s}\,]
\)
are proportional, contradicting $\operatorname{rank}[\,\tau M_{s}\ \ (1-\tau)\rho_{s}\,]=n.$
Hence 
\(
(M_{s})_{\omega\kappa}=(M_{s})_{\omega'\kappa}=0
\)
for all $\kappa\notin\{\omega,\omega'\}.$
 
Now fix $\omega\in\Omega$. Since $n\ge4$, for every $\kappa\neq\omega$,
we may choose some $\omega'\notin\{\omega,\kappa\}$. Applying the
preceding conclusion to the pair $(\omega,\omega')$ gives $(M_{s})_{\omega\kappa}=0.$
Therefore, 
\(
(M_{s})_{\omega\kappa}=0
\)
for all $\kappa\neq\omega.$ Since $\omega$ was arbitrary, $M_{s}$ is diagonal. 
\end{proof}

 \begin{lem}\label{lemma:Ms-scalar}
Suppose that, for each $s\in S$,
\[
p_{\sigma_s}(\cdot\mid s)
=
\frac{p\odot[\tau M_s\sigma_s+(1-\tau)\rho_s]}
{\|p\odot[\tau M_s\sigma_s+(1-\tau)\rho_s]\|_1}
\qquad\text{for all }\sigma\in\Sigma,
\]
where $p\in\operatorname{Int}(\Delta(\Omega))$,
$\tau\in(0,1]$, $M_s$ is diagonal, and
$\operatorname{rank}[\,\tau M_s\ \ (1-\tau)\rho_s\,]=n$.
Suppose also that Axiom~\ref{a4} holds.
Then, for each $s\in S$, there exists $k_s>0$ such that
$M_s=k_sI$.
\end{lem}

\begin{proof}
Fix $s\in S$ and write $M_s=\operatorname{Diag}(m_s)$. For
$\zeta\in\Sigma_s$, let
\(
L(\zeta):=\tau M_s\zeta+(1-\tau)\rho_s
\)
and
\(
c(\zeta):=\|p\odot L(\zeta)\|_1>0,
\)
so that the posterior at $\zeta$ is
\(
\mu(\zeta):=p\odot L(\zeta)/c(\zeta).
\)
Since $\Sigma_s$ is open, the vectors $(\zeta,1)$, $\zeta\in\Sigma_s$,
span $\mathbb{R}^{n+1}$. Hence $L(\Sigma_s)$ spans the column space of
$[\,\tau M_s\ \ (1-\tau)\rho_s\,]$, which is $\mathbb{R}^n$ by the
rank condition.

Take any $\zeta\in\Sigma_s$ and $t>0$ such that
$\zeta+t1_\Omega\in\Sigma_s$. By Axiom~\ref{a4} and the same
separating-hyperplane argument used above,
\(
\mu(\zeta+t1_\Omega)=(1-\theta)\mu(\zeta)+\theta p
\)
for some $\theta\in[0,1]$. Since $M_s1_\Omega=m_s$, we also have
\(
L(\zeta+t1_\Omega)=L(\zeta)+\tau t\,m_s.
\)
Multiplying the first equality by $c(\zeta+t1_\Omega)$, dividing
componentwise by $p$, and substituting the second gives
\(
\tau t\,m_s=(\lambda-1)L(\zeta)+\kappa1_\Omega,
\)
where
\(
\lambda:=(1-\theta)\,c(\zeta+t1_\Omega)/c(\zeta)
\)
and
\(
\kappa:=\theta\,c(\zeta+t1_\Omega)\geq0.
\)

If $m_s\notin\operatorname{span}\{1_\Omega\}$, this identity forces
$\lambda\neq1$, and hence
\(
L(\zeta)\in\operatorname{span}\{1_\Omega,m_s\}
\)
for every $\zeta\in\Sigma_s$, contradicting that $L(\Sigma_s)$ spans
$\mathbb{R}^n$. Therefore $m_s=k_s1_\Omega$ for some
$k_s\in\mathbb{R}$, that is, $M_s=k_sI$, and the identity becomes
\(
(\lambda-1)L(\zeta)=(\tau tk_s-\kappa)1_\Omega.
\)
Choose $\zeta$ with $L(\zeta)\notin\operatorname{span}\{1_\Omega\}$,
which exists because $L(\Sigma_s)$ spans $\mathbb{R}^n$. Then
$\lambda=1$ and $\tau tk_s=\kappa\geq0$, so $k_s\geq0$. Finally,
$k_s\neq0$, since otherwise $L$ would be constant on $\Sigma_s$.
Hence $k_s>0$.
\end{proof}

By Lemma \ref{lemma:Ms-scalar}, for each $s\in S$, there exists
$k_{s}>0$ such that $M_{s}=k_{s}I.$ Hence 
\[
p_{\sigma_{s}}(\cdot\mid s)=\frac{p\odot[\tau k_{s}\sigma_{s}+(1-\tau)\rho_{s}]}{\sum_{\omega\in\Omega}p_{\omega}[\tau k_{s}\sigma_{s\omega}+(1-\tau)\rho_{s\omega}]}.
\]
 
Let $\underline{k}:=\min_{s\in S}k_{s}>0$ and $\hat{\tau}:=\frac{\tau\underline{k}}{1-\tau+\tau\underline{k}}\in(0,1]$,
and define $\hat{\rho}_{s}:=\frac{\underline{k}}{k_{s}}\rho_{s}$
for each $s\in S$. Since $0<\underline{k}/k_{s}\le1$ for every $s$,
$\hat{\rho}$ is again a quasi-experiment. Moreover, 
\[
\tau k_{s}\sigma_{s}+(1-\tau)\rho_{s}=\frac{k_{s}}{\underline{k}}(1-\tau+\tau\underline{k})\bigl[\hat{\tau}\sigma_{s}+(1-\hat{\tau})\hat{\rho}_{s}\bigr].
\]
Substituting this identity and cancelling the common positive factor gives 
\[
p_{\sigma_{s}}(\cdot\mid s)=\frac{p\odot[\hat{\tau}\sigma_{s}+(1-\hat{\tau})\hat{\rho}_{s}]}{\sum_{\omega\in\Omega}p_{\omega}[\hat{\tau}\sigma_{s\omega}+(1-\hat{\tau})\hat{\rho}_{s\omega}]}.
\]
Renaming $(\hat{\tau},\hat{\rho})$ as $(\tau,\rho)$ gives the desired
representation.
 
It remains to prove uniqueness. Suppose $(\tau,\rho)$
and $(\tau',\rho')$ both satisfy \eqref{eqmain3}. Fix $s\in S$. Since
the two representations induce the same posterior at every $\sigma_{s}\in\Sigma_{s}$
and $p$ has full support, the vectors $\tau'\sigma_{s}+(1-\tau')\rho'_{s}$
and $\tau\sigma_{s}+(1-\tau)\rho_{s}$ are proportional at each $\sigma_{s}\in\Sigma_{s}$:
\(
\tau'\sigma_{s}+(1-\tau')\rho'_{s}=d\,\bigl[\tau\sigma_{s}+(1-\tau)\rho_{s}\bigr].
\)
The factor $d>0$ cannot depend on $\sigma_{s}$: both sides are affine
in $\sigma_{s}$, so comparing any two points of the open set $\Sigma_{s}$
with their midpoint forces a common factor. Matching coefficients
on $\Sigma_{s}$ then yields $\tau'=d\tau$ and $(1-\tau')\rho'_{s}=d(1-\tau)\rho_{s}$;
in particular $d=\tau'/\tau$ is the same for all $s$, and eliminating
it gives 
\begin{equation}
\frac{1-\tau'}{\tau'}\,\rho'_{s}=\frac{1-\tau}{\tau}\,\rho_{s}\quad\text{for all }s\in S.\label{eq:tradeoff}
\end{equation}
Thus \eqref{eq:tradeoff} relates any two representations:
higher trust requires a proportional enlargement of the alternative
theory. Hence $\rho$ is unique for any fixed $\tau<1$.
 
Fix one representation $(\tau,\rho)$ with $\tau<1$. By \eqref{eq:tradeoff},
the candidates are exactly $\bigl(\tau',\tfrac{\tau'(1-\tau)}{\tau(1-\tau')}\rho\bigr)$
for $\tau'\in(0,1)$, and such $\rho'$ is a quasi-experiment if and
only if, for every $\omega\in\Omega$, 
\[
\sum_{s\in S}\rho'(s\mid\omega)=\frac{\tau'(1-\tau)}{\tau(1-\tau')}\sum_{s\in S}\rho(s\mid\omega)\leq1,
\]
that is, if and only if 
\[
\tau'\;\leq\;\tau^{*}:=\Bigl(1+\tfrac{1-\tau}{\tau}\,\max_{\omega\in\Omega}\sum_{s\in S}\rho(s\mid\omega)\Bigr)^{-1}.
\]
The bound is attained: at $\tau'=\tau^{*}$ the row sum equals $1$
at the maximizing $\omega$. Hence the admissible trust parameters
form $(0,\tau^{*}\,]$, with unique maximum $\tau^{*}$
and unique associated $\bar{\rho}$ by \eqref{eq:tradeoff}.
Finally, if some representation has $\tau=1$, then the coefficient
relation forces $\tau'=1$ for every representation, and $\tau^{*}=1$,
with $\rho$ playing no role in \eqref{eqmain3}.
 
\subsection{Proof of Theorem \ref{thmmain3}}
 
By Idempotence, for some $\sigma^{*}\in\Sigma$, 
\(
p_{\sigma^{*}_{s}}(\cdot\mid s)=p
\)
for all $s\in S$.
By Theorem \ref{thmmain2}, there exist a quasi-experiment $\rho$
and $\tau\in(0,1]$ such that for all $(\sigma,s)\in\Sigma\times S$, $p_{\sigma_{s}}(\cdot\mid s)$ equals ${p\odot[\tau\sigma_{s}+(1-\tau)\rho_{s}]}\big/\allowbreak{\sum_{\omega\in\Omega}p_{\omega}[\tau\sigma_{s\omega}+(1-\tau)\rho_{s\omega}]}$. If $\tau=1$, then $\rho$ does not enter the posterior. Thus we may choose any experiment $\rho$, obtaining the desired
representation. It remains to consider $\tau\in(0,1)$.
 
Since $p_{\sigma^{*}_{s}}(\cdot\mid s)=p$, and $p\in\operatorname{Int}(\Delta(\Omega))$,
for each $s\in S$ there exists $l_{s}>0$ such that 
\(
\tau\sigma^{*}_{s}+(1-\tau)\rho_{s}=l_{s}1_{\Omega}.
\)
Summing over $s\in S$, we obtain, for every $\omega\in\Omega$, 
\(
\tau\sum_{s\in S}\sigma^{*}_{s\omega}+(1-\tau)\sum_{s\in S}\rho_{s\omega}=\sum_{s\in S}l_{s}.
\)
Since $\sigma^{*}$ is an experiment, 
\(
\sum_{s\in S}\sigma^{*}_{s\omega}=1.
\)
Hence 
\(
\sum_{s\in S}\rho_{s\omega}=\frac{\sum_{s\in S}l_{s}-\tau}{1-\tau}=:\gamma
\)
for every $\omega\in\Omega$.
Since $\rho$ is a non-zero quasi-experiment, $\gamma>0$. Therefore 
\(
\tilde{\rho}_{s}:=\frac{1}{\gamma}\rho_{s}
\)
for each $s\in S$
defines an experiment, since 
\(
\sum_{s\in S}\tilde{\rho}_{s\omega}=1
\)
for every $\omega\in\Omega$.
 
Now define 
\(
\tilde{\tau}:=\frac{\tau}{\tau+(1-\tau)\gamma}.
\)
Then $\tilde{\tau}\in(0,1]$, and for every $\sigma\in\Sigma$ and
$s\in S$, 
\[
\tau\sigma_{s}+(1-\tau)\rho_{s}=\tau\sigma_{s}+(1-\tau)\gamma\tilde{\rho}_{s}=\bigl[\tau+(1-\tau)\gamma\bigr]\bigl[\tilde{\tau}\sigma_{s}+(1-\tilde{\tau})\tilde{\rho}_{s}\bigr].
\]
Cancelling the common positive factor gives 
\[
\frac{p\odot[\tau\sigma_{s}+(1-\tau)\rho_{s}]}{\sum_{\omega\in\Omega}p_{\omega}[\tau\sigma_{s\omega}+(1-\tau)\rho_{s\omega}]}=\frac{p\odot[\tilde{\tau}\sigma_{s}+(1-\tilde{\tau})\tilde{\rho}_{s}]}{\sum_{\omega\in\Omega}p_{\omega}[\tilde{\tau}\sigma_{s\omega}+(1-\tilde{\tau})\tilde{\rho}_{s\omega}]}.
\]
Thus the experiment $\tilde{\rho}$ yields the same posterior rule. Renaming $(\tilde{\tau},\tilde{\rho})$ as $(\tau,\rho)$
gives the desired representation. Uniqueness
follows from the previous proof.
 
\subsection{Proof of Theorem \ref{thmmain1}}
 
Suppose that $\{\succsim_{0},\succsim_{\sigma,s}\}_{(\sigma,s)\in\Sigma\times S}$
satisfies Axioms $1$--$5$. By Theorem \ref{thmmain2}, there exist a quasi-experiment $\rho$
and $\tau\in(0,1]$ such that, for all $(\sigma,s)\in\Sigma\times S$, $p_{\sigma_{s}}(\cdot\mid s)$ equals ${p\odot[\tau\sigma_{s}+(1-\tau)\rho_{s}]}\big/\allowbreak{\sum_{\omega\in\Omega}p_{\omega}[\tau\sigma_{s\omega}+(1-\tau)\rho_{s\omega}]}.$ We show that necessarily $\tau=1$. Suppose, towards a contradiction, that $\tau<1$. Fix $s\in S$, and
write 
\(
b_{s}:=(1-\tau)\rho_{s}.
\)
Since $\Sigma_{s}$ has nonempty interior, choose two non-proportional
points $v,w\in\Sigma_{s}$. Since $\Sigma_{s}$ is open, for some
$\alpha\neq1$ sufficiently close to $1$, we have 
\(
\alpha v,\alpha w\in\Sigma_{s}.
\)
By Axiom \ref{a3}, $v$ and $\alpha v$ induce the same posterior. Since $p$ has full support,
the unnormalized posterior vectors are proportional:
\(
\tau\alpha v+b_{s}=\lambda_{v}(\tau v+b_{s})
\)
for some $\lambda_{v}>0$. Since $\alpha\neq1$, we cannot have $\lambda_{v}=1$.
Hence $b_{s}$ is proportional to $v$. Likewise, $b_{s}$ is proportional to $w$. This contradicts
the choice of non-collinear $v$ and $w$, unless $b_{s}=0$.
 
Thus $b_{s}=0$ for every $s\in S$. Since $b_{s}=(1-\tau)\rho_{s}$
and $\tau<1$, it follows that $\rho_{s}=0$ for every $s$, contradicting
the definition of a quasi-experiment as non-zero. Therefore $\tau=1$.

\section{Dynamic Subjective Expected Utility: Review}
 
\label{Appendix: Classic Dynamic SEU}
 
Consider a finite state space $\Omega$ and a finite signal space
$S$. An experiment is a mapping $\sigma:\Omega\rightarrow\Delta(S)$.
Let $X$ be a convex subset of a normed vector space. Let $\mathcal{F}=\{f:\Omega\rightarrow X\}$
and $\mathcal{F}_{0}=\{F:S\rightarrow\mathcal{F}\}$. As usual, for
any acts $f,h$ and event $E\subset\Omega$, the act $fEh$ pays according
to $f$ for all $\omega\in E$ and according to $h$ otherwise. Similarly
$f\{s\}H$ is a signal-act in $\mathcal{F}_{0}$ that yields $f$
if the signal $s$ obtains and otherwise yields acts according to
$H$. Such acts allow the analyst to observe how the agent would bet
on events and therefore elicit beliefs.
 
Consider an ``ex ante'' preference $\succsim_{0}$ over $\mathcal{F}_{0}$
and a family $\{\succsim_{s}\}_{s\in S}$ of ``ex post'' preferences
over $\mathcal{F}$ that admit the following Subjective Expected Utility
representations: For some non-constant affine $u:X\rightarrow\mathbb{R}$
and full-support distributions $p_{0}(\cdot)\in\Delta(\Omega\times S)$
and $p(\cdot|s)\in\Delta(\Omega)$ , $\succsim_{0}$ admits a representation
\[
U(F)=\sum_{s\in S}\sum_{\omega\in\Omega}u(F_{s}(\omega))p_{0}(\omega,s),\qquad F\in\mathcal{F}_{0},
\]
and for each $s\in S$, $\succsim_{s}$ admits a representation 
\[
U_{s}(f)=\sum_{\omega\in\Omega}u(f(\omega))p(\omega\mid s),\qquad f\in\mathcal{F}.
\]
To establish a connection between $p_{0}(\cdot)$ and each $p(\cdot|s)$,
consider:
 
\begin{defn}[Dynamic Consistency]\label{DC} For all $f,g\in\mathcal{F}$,
$H\in\mathcal{F}_{0}$ and $s\in S$, 
\[
f\{s\}H\succsim_{0}g\{s\}H\iff f\succsim_{s}g.
\]
\end{defn}
 
The following result is the standard justification for the use of
Bayesian updating in economics.\footnote{See \citet{Ghirardato2002} for the earliest formalization of this folk
wisdom in the context of partitional information.} The result states that any agent satisfying Dynamic Consistency can
be represented as a Bayesian updater who subjectively perceives signals
to be generated by some $\pi$. The subjective experiment $\pi$ is
uniquely identified. This implies that if the analyst observes that
an agent appears to update in a non-Bayesian fashion with respect
to some objective $\sigma$, then they need only check Dynamic Consistency
to determine whether the agent can be represented as a Bayesian with
some (misspecified) subjective experiment $\pi$.
 
\begin{thm}\label{thm-classic} $\succsim_{0}$ and $\{\succsim_{s}\}_{s\in S}$
satisfy Dynamic Consistency if and only if there exists a subjective
experiment $\pi$ such that for all $s,\omega$, $\pi(s|\omega)=\frac{p_{0}(\omega,s)}{p_{0}(\omega)}$
and 
\[
p(\omega\mid s)=\frac{\pi(s|\omega)p_{0}(\omega)}{\sum_{\omega'}\pi(s|\omega')p_{0}(\omega')}
\]
Moreover, the subjective experiment is unique. \end{thm} 
\begin{proof}
Write $U(F)=\sum_{s\in S}\left[\sum_{\omega\in\Omega}u(F_{s}(\omega))p_{0}(\omega|s)\right]p_{0}(s)$,
where $p_{0}(s)=\sum_{\omega\in\Omega}p_{0}(\omega,s)$ is the marginal
over $S$ and $p_{0}(\omega|s):=\frac{p_{0}(\omega,s)}{p_{0}(s)}$
is the conditional probability. Then by Dynamic Consistency, for all
$s,f,g$, 
\[
\sum_{\omega\in\Omega}u(f(\omega))p_{0}(\omega|s)\ge\sum_{\omega\in\Omega}u(g(\omega))p_{0}(\omega|s)\iff f\{s\}H\succsim_{0}g\{s\}H
\]
\[
\iff f\succsim_{s}g\iff\sum_{\omega\in\Omega}u(f(\omega))p(\omega|s)\ge\sum_{\omega\in\Omega}u(g(\omega))p(\omega|s).
\]
By the uniqueness of the SEU representation, we obtain $p(\omega|s)=\frac{p_{0}(\omega,s)}{p_{0}(s)}$.
Define a subjective experiment by $\pi(s|\omega):=\frac{p_{0}(\omega,s)}{p_{0}(\omega)}.$
Then the posterior is a Bayesian update with respect to $\pi$: 
\[
p(\omega|s)=\frac{p_{0}(\omega,s)}{p_{0}(s)}=\frac{p_{0}(\omega,s)/p_{0}(\omega)}{\sum_{\omega'}p_{0}(\omega',s)}p_{0}(\omega)
\]
\[
=\frac{p_{0}(\omega,s)/p_{0}(\omega)}{\sum_{\omega'}\pi(s|\omega')p_{0}(\omega')}p_{0}(\omega)=\frac{\pi(s|\omega)p_{0}(\omega)}{\sum_{\omega'}\pi(s|\omega')p_{0}(\omega')}.
\]
\end{proof}

\section{Relaxing the Boundary Condition}
\label{sec:relaxing-boundary}
 
We prove the counterpart of Theorem~\ref{thmmain2} where the requirement
``$0\in\overline{\Sigma_{s}}$ for each $s\in S$'' on $\Sigma$
is replaced with the following condition on preferences.
 
\begin{ax}[Compensated Likelihood Increase]
\label{ax:cli-boundary}
For every $s\in S$, $\omega\in\Omega$, $\sigma,\sigma'\in\Sigma$,
and $\lambda\in(0,1)$ satisfying
$\sigma_{s}(\omega)=\lambda\sigma'_{s}(\omega)$
and $\sigma'_{s}(\psi)=\sigma_{s}(\psi)$ for all $\psi\neq\omega$,
and every $x,y,z\in X$ with $x\succsim_{0}y$,
\[
x_{\{\omega\}}y\sim_{\sigma,s}z
\quad\Longrightarrow\quad
\bigl[\lambda x+(1-\lambda)z\bigr]_{\{\omega\}}y
\precsim_{\sigma',s}z.
\]
\end{ax}
 
The axiom starts from a bet that is indifferent to the consequence
$z$. The likelihood of the observed signal in state $\omega$ is
then increased by a factor of $1/\lambda$, while the bet's utility
gain relative to $z$ in that state is reduced by a factor of $\lambda$.
For a Bayesian agent who fully trusts the stated experiment, these
changes exactly offset each other. Under limited trust, however, part
of the perceived likelihood comes from an alternative model of signal
generation and remains unchanged. The increase in the posterior odds
of $\omega$ is therefore attenuated, so the reduction in the gain
is sufficient to offset the likelihood increase. The agent consequently
weakly prefers the sure consequence $z$ to the adjusted bet.
 
\begin{thm}\label{thm:relaxing-boundary}
Maintain Axiom~\ref{axiom:seu} and $|\Omega|\geq4$.
Suppose $\Sigma$ is nonempty and convex, and each
$\Sigma_{s}$ is open in $\mathbb{R}^{\Omega}$.
Then $\{\succsim_{0},\succsim_{\sigma,s}\}_{(\sigma,s)\in\Sigma\times S}$
satisfies Axioms~\ref{a1}, \ref{a2}, \ref{a4}, and~\ref{a5},
together with Axiom~\ref{ax:cli-boundary},
if and only if there exist $\tau\in(0,1]$ and a quasi-experiment
$\rho$ such that
\begin{equation}
p_{\sigma}(\omega\mid s)
=
\frac{
p(\omega)\bigl[\tau\sigma(s\mid\omega)
+(1-\tau)\rho(s\mid\omega)\bigr]
}{
\sum_{\psi\in\Omega}p(\psi)
\bigl[\tau\sigma(s\mid\psi)
+(1-\tau)\rho(s\mid\psi)\bigr]
}
\label{eq:boundary-representation}
\end{equation}
for every $(\sigma,s)\in\Sigma\times S$ and $\omega\in\Omega$.
\end{thm}
 
\begin{proof}
Suppose first that the stated axioms hold. The affine-representation
argument in the proof of Theorem~\ref{thmmain4} and the arguments
in Lemmas~\ref{lemma:Ms-diagonal} and~\ref{lemma:Ms-scalar}
can be carried out before imposing $b_{s}\geq0$.
Specifically, retain the raw affine numerator $A_{s}\zeta+b_{s}$
and orient its denominator to be positive on $\Sigma_{s}$.
Richness gives
$\operatorname{rank}[A_{s}\ b_{s}]=|\Omega|$
directly. Lemma~\ref{lemma:Ms-diagonal} then gives diagonality,
and the argument in Lemma~\ref{lemma:Ms-scalar} gives
$A_{s}=k_{s}D_{p}$ with $k_{s}>0$, where
$D_{p}=\operatorname{diag}(p)$.
Thus we obtain
\begin{equation}
\mu_{s}(\zeta):=p_{\sigma}(\cdot\mid s)
=\frac{k_{s}D_{p}\zeta+b_{s}}{D_{s}(\zeta)},
\qquad
D_{s}(\zeta):=\mathbf{1}^{\top}(k_{s}D_{p}\zeta+b_{s})>0,
\qquad \zeta=\sigma_{s},
\label{eq:boundary-signed-form}
\end{equation}
without restricting the sign of $b_{s}$.
The construction of a quasi-experiment is postponed until
nonnegativity of $b_{s}$ has been established.
Also, notice that the argument establishing $k_{s}>0$ in
Lemma~\ref{lemma:Ms-scalar} does not require
$0\in\overline{\Sigma_{s}}$.\footnote{
To see why, notice that once $A_{s}=k_{s}D_{p}$ has been obtained,
for any sufficiently small $t>0$ with
$\zeta+t\mathbf{1}\in\Sigma_{s}$,
\[
\mu_{s}(\zeta+t\mathbf{1})
=
\frac{D_{s}(\zeta)}{D_{s}(\zeta)+tk_{s}}\mu_{s}(\zeta)
+
\frac{tk_{s}}{D_{s}(\zeta)+tk_{s}}p.
\]
The denominator is positive. At a point where
$\mu_{s}(\zeta)\neq p$, which exists by Richness,
Axiom~\ref{a4} requires the coefficient on $p$ to be positive.
Hence $k_{s}>0$ without taking likelihoods to zero.
}
 
We next prove $b_{s}\geq0$ using Axiom~\ref{ax:cli-boundary}.
Fix $s$, $\omega$, and any $\sigma\in\Sigma$, and put
$\zeta=\sigma_{s}$.
Every numerator $k_{s}p_{\psi}\zeta_{\psi}+b_{s,\psi}$
is nonnegative. It is in fact positive: if it were zero,
openness would allow a small reduction of $\zeta_{\psi}$,
making that numerator negative because $k_{s}p_{\psi}>0$.
Therefore
\(
q:=[\mu_{s}(\zeta)]_{\omega}\in(0,1).
\)
Choose $x,y\in X$ with $u(x)>u(y)$ and set
\(
z:=qx+(1-q)y.
\)
Then $x_{\{\omega\}}y\sim_{\sigma,s}z$ and
$x\succsim_{0}z\succsim_{0}y$.
Let
\(
a:=u(x)-u(z)>0
\)
and
\(
d:=u(z)-u(y)>0.
\)
Initial indifference gives
\begin{equation}
(k_{s}p_{\omega}\zeta_{\omega}+b_{s,\omega})a
=
\sum_{\psi\neq\omega}
(k_{s}p_{\psi}\zeta_{\psi}+b_{s,\psi})d.
\label{eq:boundary-indifference}
\end{equation}
Choose $\lambda\in(0,1)$ sufficiently close to one that
the vector $\zeta'$ defined by
$\lambda\zeta'_{\omega}=\zeta_{\omega}$ and
$\zeta'_{\psi}=\zeta_{\psi}$ for $\psi\neq\omega$
belongs to $\Sigma_{s}$.
Choose $\sigma'\in\Sigma$ with $\sigma'_{s}=\zeta'$,
and define
\(
g:=\bigl[\lambda x+(1-\lambda)z\bigr]_{\{\omega\}}y.
\)
By affinity of $u$ and \eqref{eq:boundary-indifference},
\begin{align}
U_{\sigma',s}(g)-u(z)
&=
\frac{
(k_{s}p_{\omega}\zeta'_{\omega}+b_{s,\omega})\lambda a
-\sum_{\psi\neq\omega}
(k_{s}p_{\psi}\zeta_{\psi}+b_{s,\psi})d
}{
D_{s}(\zeta')
}
\nonumber\\
&=
-\frac{(1-\lambda)b_{s,\omega}a}{D_{s}(\zeta')}.
\label{eq:boundary-compensation}
\end{align}
The axiom makes this expression nonpositive.
Since $1-\lambda>0$, $a>0$, and $D_{s}(\zeta')>0$,
we obtain $b_{s,\omega}\geq0$.
This holds for every $s$ and $\omega$.
The normalization used in the proof of
Theorem~\ref{thmmain2} then yields a common
$\tau\in(0,1]$ and a quasi-experiment $\rho$
satisfying \eqref{eq:boundary-representation}.
 
Conversely, suppose \eqref{eq:boundary-representation} holds.
The verification of Axioms~\ref{a1}, \ref{a2}, \ref{a4},
and~\ref{a5} is the same as in the main-text characterization
and does not use the boundary condition.
We therefore verify only Axiom~\ref{ax:cli-boundary}.
Set
\(
\delta_{s}:=\frac{1-\tau}{\tau}\rho_{s}
\)
and
\(
Z_{s}(\zeta):=p\cdot(\zeta+\delta_{s})>0.
\)
 
For any comparison in Axiom~\ref{ax:cli-boundary},
the same calculation as in \eqref{eq:boundary-compensation},
with $k_{s}=1$ and $b_{s}=D_{p}\delta_{s}$, gives
\[
U_{\sigma',s}(g)-u(z)
=
-\frac{
(1-\lambda)p_{\omega}\delta_{s,\omega}[u(x)-u(z)]
}{
Z_{s}(\sigma'_{s})
}
\leq0.
\]
Thus the axiom holds. In particular, if $\tau=1$,
this difference is zero, so the Bayesian case is included.
\end{proof}

\section{Proof of Propositions \ref{prop:under}-\ref{prop: confreading1}}
 
\subsection{Proof of Proposition \ref{prop:under}}
 
The ``if'' side is trivial, and we focus on the ``only if'' side.
Fix $s\in S$ and assume $\rho_{s}\neq0$; otherwise the claim is
immediate. By a standard separating hyperplane argument, Axiom~\ref{a:under}
implies that, whenever $\sigma_{s}=\alpha\sigma'_{s}$ with $\alpha\in(0,1)$,
\[
p_{\sigma}(\cdot\mid s)\in\operatorname{co}\{p,\,p_{\sigma'}(\cdot\mid s)\}.
\]
 
Suppose, towards a contradiction, that $p^{*}_{\rho}(\cdot\mid s)\neq p$.
Since $\Sigma_{s}$ is open, the Bayesian posteriors $p^{*}_{\sigma'}(\cdot\mid s)$,
as $\sigma'_{s}$ ranges over $\Sigma_{s}$, form an open subset of
$\Delta(\Omega)$, whose dimension is $|\Omega|-1\geq3$. We may therefore
choose $\sigma'\in\Sigma$ such that $p$, $p^{*}_{\rho}(\cdot\mid s)$,
and $p^{*}_{\sigma'}(\cdot\mid s)$ are not collinear, and, again
by openness, $\sigma\in\Sigma$ with $\sigma_{s}=\alpha\sigma'_{s}$
for some $\alpha\in(0,1)$.
 
By \eqref{eqmain3}, $p_{\sigma'}(\cdot\mid s)$ and $p_{\sigma}(\cdot\mid s)$
lie on the segment joining $p^{*}_{\sigma'}(\cdot\mid s)$ and $p^{*}_{\rho}(\cdot\mid s)$,
and they are distinct because $\alpha\neq1$ and $\tau<1$. The segment
$\operatorname{co}\{p,\,p_{\sigma'}(\cdot\mid s)\}$ thus contains
two distinct points of the line through $p^{*}_{\sigma'}(\cdot\mid s)$
and $p^{*}_{\rho}(\cdot\mid s)$, so $p$ lies on this line as well---a
contradiction.
 
Hence $p^{*}_{\rho}(\cdot\mid s)=p$, so that $\rho_{s}$ is proportional
to $1_{\Omega}$. As $s$ was arbitrary, $\rho$ is uninformative.
\qed
 
\subsection{Proof of Proposition \ref{prop-conservatism}}
\begin{proof}
Suppose Axiom~\ref{a:under} holds. If $\tau=1$ then $p_{\sigma}(\cdot|s)=p^{*}_{\sigma}(\cdot|s)$
for all $(\sigma,s)$ and conservatism is immediate, so assume $\tau<1$.
By Proposition~\ref{prop:under}, $\rho$ is uninformative. For any
$s\in S$ with $\rho_{s}\neq0$, given 
\[
p_{\sigma}(\omega|s)=\tau(\sigma|s)p^{*}_{\sigma}(\omega|s)+(1-\tau(\sigma|s))p^{*}_{\rho}(\omega|s),\qquad\tau(\sigma|s)=\frac{\tau\sigma_{p}(s)}{\tau\sigma_{p}(s)+(1-\tau)\rho_{p}(s)},
\]
we see that $p_{\sigma}(\cdot|s)$ is a convex combination of $p^{*}_{\sigma}(\cdot|s)$
and $p(\cdot)$ because $p^{*}_{\rho}(\cdot\mid s)=p(\cdot)$. On
the other hand, if $\rho_{s}=0$, then $p_{\sigma}(\cdot|s)=p^{*}_{\sigma}(\cdot|s)$.
In either case $p_{\sigma}(\cdot|s)$ is a convex combination of $p^{*}_{\sigma}(\cdot|s)$
and $p(\cdot)$, and it follows that $p_{\sigma}$ underreacts relative
to $p^{*}_{\sigma}$ for all $s$.
 
Conversely, suppose that the agent's beliefs exhibit conservatism.
If $\tau=1$ then $p_{\sigma}(\cdot|s)=p^{*}_{\sigma}(\cdot|s)$ for
all $(\sigma,s)$, so Axiom~\ref{a:under} holds trivially; assume
henceforth $\tau<1$. Let $\underline{\sigma}_{s}\in\Sigma_{s}$ be
a constant likelihood vector, with corresponding experiment $\underline{\sigma}\in\Sigma$.
Since $p^{*}_{\underline{\sigma}}(\cdot|s)=p(\cdot)$, conservatism
implies that for all $\omega$, 
\[
p_{\underline{\sigma}}(\omega|s)=p(\omega).
\]
Fix any $s\in S$. We want to establish that $\rho_{s}$ is a constant
vector, which trivially holds true if $\rho_{s}=0$. If $\rho_{s}\neq0$,
then 
\[
p_{\underline{\sigma}}(\omega|s)=\tau(\underline{\sigma}|s)p^{*}_{\underline{\sigma}}(\omega|s)+(1-\tau(\underline{\sigma}|s))p^{*}_{\rho}(\omega|s)
\]
\[
\implies p(\omega)=\tau(\underline{\sigma}|s)p(\omega)+(1-\tau(\underline{\sigma}|s))p^{*}_{\rho}(\omega|s)
\]
\[
\implies(1-\tau(\underline{\sigma}|s))p^{*}_{\rho}(\omega|s)=(1-\tau(\underline{\sigma}|s))p(\omega).
\]
Since $\rho_{s}\neq0$, it holds that $\tau(\underline{\sigma}\mid s)<1$.
It follows that $p^{*}_{\rho}(\omega|s)=p(\omega)$ for all $\omega$,
and therefore $\rho_{s}$ is a constant vector. Hence $\rho$ is uninformative,
and Axiom~\ref{a:under} holds by Proposition~\ref{prop:under}. 
\end{proof}

\subsection{Proof of Proposition \ref{prop: overreaction0}}

By the separating hyperplane argument in the proof of Proposition~\ref{prop:under},
the agent overreacts to surprises at $\sigma$ if and only if, for
all $s\in S$ and all $\alpha\in(0,1)$ with $\alpha\sigma_{s}\in\Sigma_{s}$,
\[
p_{\sigma}(\cdot\mid s)\in\operatorname{co}\{p,\,p_{\sigma'}(\cdot\mid s)\},\qquad\sigma'_{s}=\alpha\sigma_{s}.
\]
Cancelling $p$ componentwise, this holds if and only if there exist
$a,b\geq0$ such that 
\begin{align}
\tau\sigma_{s}+(1-\tau)\rho_{s}=a1_{\Omega}+b\bigl[\tau\alpha\sigma_{s}+(1-\tau)\rho_{s}\bigr].\label{eq:overcone}
\end{align}
 
We prove the only if direction of the second claim; the if direction
follows immediately from \eqref{eq:overcone} and is omitted. Fix
$s\in S$. Suppose first that $\sigma_{s}$ is a constant vector.
Then $p^{*}_{\sigma}(\omega\mid s)=p(\omega)$ for all $\omega$,
so no backfire at $\sigma$ requires $p_{\sigma}(\omega\mid s)=p(\omega)$
for all $\omega$. Hence $\tau\sigma_{s}+(1-\tau)\rho_{s}$ is constant,
and therefore $\rho_{s}$ is constant and admits the stated form.
 
Suppose now that $\sigma_{s}$ is not constant. If $1_{\Omega}$,
$\sigma_{s}$, and $\rho_{s}$ were linearly independent, matching
coefficients in \eqref{eq:overcone} would give $b=1/\alpha$ from
$\sigma_{s}$ and $b=1$ from $\rho_{s}$, contradicting $\alpha\in(0,1)$.
Hence 
\[
\rho_{s}=\alpha_{s}\sigma_{s}-\beta_{s}1_{\Omega}
\]
for some scalars $\alpha_{s},\beta_{s}$, so that 
\begin{align}
\tau\sigma_{s}+(1-\tau)\rho_{s}=\bigl[\tau+(1-\tau)\alpha_{s}\bigr]\sigma_{s}-(1-\tau)\beta_{s}1_{\Omega}.\label{eqmix}
\end{align}
 
We claim that 
\[
\tau+(1-\tau)\alpha_{s}>0.
\]
Since $\sigma_{s}$ is not constant, there is $\omega$ such that
$\sigma(s\mid\omega)>\sigma_{p}(s)$, and hence
$p^{*}_{\sigma}(\omega\mid s)>p(\omega)$.
If $\tau+(1-\tau)\alpha_{s}\leq0$, then \eqref{eqmix},
with $\pi:=\tau\sigma+(1-\tau)\rho$, implies
\[
\pi(s\mid\omega)\leq\pi_{p}(s),
\]
and therefore $p_{\sigma}(\omega\mid s)\leq p(\omega)$,
contradicting no backfire at $\sigma$.
 
Since $\sigma_{s}$ and $1_{\Omega}$ are linearly independent, substituting
\eqref{eqmix} into \eqref{eq:overcone} and matching the coefficients
of $\sigma_{s}$ and $1_{\Omega}$ gives 
\[
b=\frac{\tau+(1-\tau)\alpha_{s}}{\tau\alpha+(1-\tau)\alpha_{s}}>1,\qquad a=(b-1)(1-\tau)\beta_{s}.
\]
Thus $a\geq0$ implies $\beta_{s}\geq0$. Finally, if $\alpha_{s}<0$,
then $\rho_{s}=\alpha_{s}\sigma_{s}-\beta_{s}1_{\Omega}$ has a negative
component, contradicting $\rho_{s}\geq0$. Hence $\alpha_{s}\geq0$.
 
For the first claim, suppose Axiom~\ref{a:over} holds. Then the
agent overreacts to surprises at every $\sigma\in\Sigma$. Fix $s\in S$.
Since $\Sigma_{s}$ is open and $|\Omega|\geq4$, choose $\sigma,\tilde{\sigma}\in\Sigma$
such that $1_{\Omega}$, $\sigma_{s}$, and $\tilde{\sigma}_{s}$
are linearly independent. The linear-independence argument above implies
that $\rho_{s}$ lies both in $\operatorname{span}\{\sigma_{s},1_{\Omega}\}$
and in $\operatorname{span}\{\tilde{\sigma}_{s},1_{\Omega}\}$. Hence
\[
\rho_{s}\in\operatorname{span}\{\sigma_{s},1_{\Omega}\}\cap\operatorname{span}\{\tilde{\sigma}_{s},1_{\Omega}\}=\operatorname{span}\{1_{\Omega}\}.
\]
Since $\rho_{s}\geq0$, we have $\rho_{s}=\lambda1_{\Omega}$ for
some $\lambda\geq0$. Substituting this into \eqref{eq:overcone}
and matching coefficients gives 
\[
b=\frac{1}{\alpha}>1,\qquad a=(1-b)(1-\tau)\lambda.
\]
Since $a\geq0$ and $\tau<1$, it follows that $\lambda=0$. Thus
$\rho_{s}=0$ for every $s\in S$, contradicting $\rho\neq0$. \qed 
 
\subsection{Proof of Proposition \ref{prop: overreaction}}

If $\tau=1$, the conclusion is immediate. Hence assume $\tau<1$. By Proposition~\ref{prop: overreaction0},
$\rho_{s}=\alpha_{s}\sigma_{s}-\beta_{s}1_{\Omega}$
with $\alpha_{s}\geq0$, $\beta_{s}\geq0$, so that 
\[
\tau\sigma_{s}+(1-\tau)\rho_{s}=\bigl[\tau+(1-\tau)\alpha_{s}\bigr]\bigl[\sigma_{s}-c_{s}1_{\Omega}\bigr],\qquad c_{s}:=\frac{(1-\tau)\beta_{s}}{\tau+(1-\tau)\alpha_{s}}\geq0.
\]
Hence $p_{\sigma}(\cdot\mid s)$ is the Bayesian posterior with respect
to $\sigma_{s}-c_{s}1_{\Omega}$, that is, 
\[
p_{\sigma}(\omega\mid s)=\frac{p(\omega)\bigl[\sigma(s\mid\omega)-c_{s}\bigr]}{\sigma_{p}(s)-c_{s}},\qquad p^{*}_{\sigma}(\omega\mid s)=\frac{p(\omega)\sigma(s\mid\omega)}{\sigma_{p}(s)}.
\]
Subtracting the two expressions, 
\[
p_{\sigma}(\omega\mid s)-p^{*}_{\sigma}(\omega\mid s)=\frac{p(\omega)\,c_{s}\bigl[\sigma(s\mid\omega)-\sigma_{p}(s)\bigr]}{\sigma_{p}(s)\bigl[\sigma_{p}(s)-c_{s}\bigr]},
\]
and 
\[
p^{*}_{\sigma}(\omega\mid s)-p(\omega)=\frac{p(\omega)\bigl[\sigma(s\mid\omega)-\sigma_{p}(s)\bigr]}{\sigma_{p}(s)}.
\]
Since $c_{s}\geq0$ and $\sigma_{p}(s)-c_{s}>0$, the two differences
have the same sign. Therefore $p^{*}_{\sigma}(\omega\mid s)\geq p(\omega)$
implies $p_{\sigma}(\omega\mid s)\geq p^{*}_{\sigma}(\omega\mid s)$,
and $p^{*}_{\sigma}(\omega\mid s)\leq p(\omega)$ implies $p_{\sigma}(\omega\mid s)\leq p^{*}_{\sigma}(\omega\mid s)$,
which is the definition of overreaction at $\sigma$. \qed

\subsection{Proof of Proposition \ref{propconfirmation}}
\begin{proof}
Let $\pi=\tau\sigma+(1-\tau)\rho$. For an experiment $e$ and a signal
$s$, write $e_{p}(s)=\sum_{\omega}p(\omega)e(s\mid\omega)$ for the
ex ante probability of $s$ under $e$. Since signals are drawn from
$\sigma$, the agent's average posterior on $\omega$ is
\[
\bar{q}_{p}(\omega)=\sum_{s}p^{*}_{\pi}(\omega\mid s)\,\sigma_{p}(s).
\]
Bayesian plausibility holds with respect to $\pi$, so
$p(\omega)=\sum_{s}p^{*}_{\pi}(\omega\mid s)\,\pi_{p}(s)$,
and hence
\[
\bar{q}_{p}(\omega)-p(\omega)
=
\sum_{s}p^{*}_{\pi}(\omega\mid s)
\bigl[\sigma_{p}(s)-\pi_{p}(s)\bigr]
=
(1-\tau)\sum_{s}p^{*}_{\pi}(\omega\mid s)
\bigl[\sigma_{p}(s)-\rho_{p}(s)\bigr].
\]
Since $\sigma_{p}$ and $\rho_{p}$ are both distributions over
$S=\{s_{1},s_{2}\}$, we have
$\sigma_{p}(s_{1})-\rho_{p}(s_{1})
=-[\sigma_{p}(s_{2})-\rho_{p}(s_{2})]$,
so that
\begin{align}
\bar{q}_{p}(\omega)-p(\omega)
=
(1-\tau)
\bigl[\sigma_{p}(s_{1})-\rho_{p}(s_{1})\bigr]
\bigl[p^{*}_{\pi}(\omega\mid s_{1})
-p^{*}_{\pi}(\omega\mid s_{2})\bigr].
\label{eq:wedge}
\end{align}
Now consider $\omega=\omega_{1}$. Because there is no backfire at
$\sigma$, the last bracket is strictly positive, and $\tau<1$
makes the first factor strictly positive. Therefore
\begin{align}
\bar{q}_{p}(\omega_{1})>p(\omega_{1})
\iff
\Phi(p):=\sigma_{p}(s_{1})-\rho_{p}(s_{1})>0.
\label{eq:sign}
\end{align}
For $\omega_{2}$, the corresponding condition is $\Phi(p)<0$.
 
Write $a:=\sigma(s_{1}\mid\omega_{1})-\rho(s_{1}\mid\omega_{1})$
and $b:=\sigma(s_{1}\mid\omega_{2})-\rho(s_{1}\mid\omega_{2})$. Then
\[
\Phi(p)=p(\omega_{1})\,a+\bigl(1-p(\omega_{1})\bigr)\,b,
\]
which is affine in $p(\omega_{1})$, with $\Phi\to b$ as
$p(\omega_{1})\to0$ and $\Phi\to a$ as $p(\omega_{1})\to1$.
All three claims follow from the sign pattern of $(a,b)$,
using \eqref{eq:sign} and the fact that $p(\omega_{1})$ ranges
over the open interval $(0,1)$.
 
\emph{$\omega_{1}$-biased updating.}
$\Phi(p)>0$ for every $p(\omega_{1})\in(0,1)$ if and only if
$a\geq0$ and $b\geq0$ with at least one inequality strict.
If both are nonnegative and one is strict, then $\Phi$ is
a convex combination of $a$ and $b$ with strictly positive weights,
hence strictly positive on $(0,1)$. Conversely, if $a<0$ then
$\Phi(p)<0$ for $p(\omega_{1})$ close to $1$, and if $b<0$ then
$\Phi(p)<0$ for $p(\omega_{1})$ close to $0$; and if $a=b=0$
then $\Phi\equiv0$. This gives the condition in Proposition~\ref{propconfirmation}(\ref{prop:bias-state}). The same argument
applied to $\omega_{2}$ in place of $\omega_{1}$ characterizes
$\omega_{2}$-biased updating by the reverse inequalities
$a\leq0$ and $b\leq0$, with at least one strict.
 
\emph{Confirmatory updating.}
$\Phi$ is affine and $\Phi(p)>0$ if and only if $p(\omega_{1})$
exceeds some $t\in(0,1)$ precisely when $\Phi$ is strictly
increasing in $p(\omega_{1})$ and changes sign inside $(0,1)$,
that is, when $b<0<a$; the cutoff is then $t=b/(b-a)\in(0,1)$.
The common threshold $\theta=\max\{t,1-t\}$ satisfies the definition
for both states. Conversely, the definition requires $\Phi$ to be
positive near $p(\omega_{1})=1$ and negative near $p(\omega_{1})=0$,
which forces $b<0<a$.
Now $b<0$ reads
$\sigma(s_{1}\mid\omega_{2})<\rho(s_{1}\mid\omega_{2})$,
which by $\sigma(s_{1}\mid\omega_{2})=1-\sigma(s_{2}\mid\omega_{2})$
and $\rho(s_{1}\mid\omega_{2})=1-\rho(s_{2}\mid\omega_{2})$
is equivalent to
$\sigma(s_{2}\mid\omega_{2})>\rho(s_{2}\mid\omega_{2})$.
Together with $a>0$ this gives the condition in Proposition~\ref{propconfirmation}(\ref{prop:bias-confirmatory}).
 
\emph{Average base-rate neglect.}
Symmetrically, $\Phi(p)>0$ if and only if $p(\omega_{1})$ falls
below some $t\in(0,1)$ precisely when $a<0<b$, again with
$t=b/(b-a)$. A common threshold for both states is
$\theta=\min\{t,1-t\}$.
Here $a<0$ reads
$\sigma(s_{1}\mid\omega_{1})<\rho(s_{1}\mid\omega_{1})$
and $b>0$ is equivalent to
$\sigma(s_{2}\mid\omega_{2})<\rho(s_{2}\mid\omega_{2})$,
which gives the condition in Proposition~\ref{propconfirmation}(\ref{prop:bias-neglect}).
 
Finally, the four sign patterns just described exhaust the
possibilities apart from $a=b=0$, in which case $\Phi\equiv0$
and the martingale property holds at every prior.
\end{proof}
 
\subsection{Proof of Proposition \ref{prop: confreading1}}
 
Fix $p\in(0,1)$, $\sigma\in\Sigma^{1}$, and $\tau\in(0,1]$.
Since $p$ is fixed in what follows, we suppress the dependence of
$\rho$ on the prior $p$ in the notation. It suffices to prove that, at $p,\sigma$, the agent exhibits $\omega_i$-confirmatory reading
$p_{\sigma}(\omega_i|s_1,s_2)>p(\omega_i)$ iff
\[
\left|
\tau\sigma(s_1|\omega_i)+(1-\tau)\rho(s_1|\omega_i)-\frac12
\right|
<
\left|
\tau\sigma(s_1|\omega_{-i})+(1-\tau)\rho(s_1|\omega_{-i})-\frac12
\right|.
\]
 
\begin{proof}
Since $\sigma$ is an interior experiment and $\tau>0$, all
subjective signal likelihoods are strictly positive.
By conditional independence of the two signals given the state,
$\frac{p_{\sigma}(\omega_1|s_1,s_2)}
{p_{\sigma}(\omega_2|s_1,s_2)}$
is given by the expression
\[
\frac{
\left(\tau\sigma(s_2|\omega_1)+(1-\tau)\rho(s_2|\omega_1)\right)
\left(\tau\sigma(s_1|\omega_1)+(1-\tau)\rho(s_1|\omega_1)\right)
}{
\left(\tau\sigma(s_2|\omega_2)+(1-\tau)\rho(s_2|\omega_2)\right)
\left(\tau\sigma(s_1|\omega_2)+(1-\tau)\rho(s_1|\omega_2)\right)
}
\frac{p(\omega_1)}{p(\omega_2)}.
\]
We compute that $\omega_1$-confirmatory reading is equivalent to
\[
\frac{p_{\sigma}(\omega_1|s_1,s_2)}
     {p_{\sigma}(\omega_2|s_1,s_2)}
>
\frac{p(\omega_1)}{p(\omega_2)},
\]
or, equivalently,
\[
\begin{aligned}
&
\left(\tau\sigma(s_2|\omega_1)+(1-\tau)\rho(s_2|\omega_1)\right)
\left(\tau\sigma(s_1|\omega_1)+(1-\tau)\rho(s_1|\omega_1)\right)
\\
&\quad >
\left(\tau\sigma(s_2|\omega_2)+(1-\tau)\rho(s_2|\omega_2)\right)
\left(\tau\sigma(s_1|\omega_2)+(1-\tau)\rho(s_1|\omega_2)\right).
\end{aligned}
\]
Since $\sigma$ and $\rho$ are experiments on $\{s_1,s_2\}$,
for each $\omega\in\Omega$ we have
\[
\tau\sigma(s_2|\omega)+(1-\tau)\rho(s_2|\omega)
=
1-\left(\tau\sigma(s_1|\omega)+(1-\tau)\rho(s_1|\omega)\right).
\]
Since the function $r\mapsto(1-r)r$ defined on the unit interval
is hump-shaped and symmetric around its maximum at $r=\frac12$,
the above condition is equivalent to
\[
\left|
\tau\sigma(s_1|\omega_1)+(1-\tau)\rho(s_1|\omega_1)-\frac12
\right|
<
\left|
\tau\sigma(s_1|\omega_2)+(1-\tau)\rho(s_1|\omega_2)-\frac12
\right|.
\]
An analogous result holds for $\omega_2$-confirmatory reading.
\end{proof}

\end{document}